\documentclass[12pt, draftclsnofoot, onecolumn]{IEEEtran}

\normalsize

\IEEEoverridecommandlockouts

\usepackage{amsfonts,amsmath,amssymb,amsthm}
\usepackage{bbm,bm}
\usepackage{graphicx,epsfig,color}
\usepackage{algorithm,algorithmic}
\usepackage{comment}
\usepackage{epstopdf}
\usepackage[tight,footnotesize]{subfigure}
\usepackage{enumerate}

\newtheorem{thm}{Theorem}

\newtheorem{lem}{Lemma}[section]

\newcommand{\red}{\textcolor{red}}

\newcommand{\eg}{{\it e.g., }}

\newcommand{\ie}{{\it i.e., }}
\newcommand{\E}{{\mathbb E }}

\newcommand*{\Scale}[2][4]{\scalebox{#1}{$#2$}}%

\ifCLASSINFOpdf
\else
\fi

\begin{document}
%
\title{Optimal Control of Wireless Computing Networks}
%
%
%

\author{Hao~Feng,~\IEEEmembership{Student Member,~IEEE,}
        Jaime~Llorca,~\IEEEmembership{Member,~IEEE,}
        Antonia~M.~Tulino,~\IEEEmembership{Fellow,~IEEE,}
        and~Andreas~F.~Molisch,~\IEEEmembership{Fellow,~IEEE}}
\maketitle

\begin{abstract}
\emph{Augmented information (AgI) services} allow users to consume information that results from the execution of a chain of service functions that process source information to create real-time augmented value. Applications include real-time analysis of remote sensing data, real-time computer vision, personalized video streaming, and augmented reality, among others.
We consider the problem of optimal distribution of AgI services over a wireless computing network, in which nodes are equipped with both communication and computing resources. We characterize the wireless computing network capacity region and design a joint flow scheduling and resource allocation algorithm that stabilizes the underlying queuing system while achieving a network cost arbitrarily close to the minimum, with a tradeoff in network delay.
Our solution captures the unique chaining and flow scaling aspects of AgI services, while exploiting the use of the \emph{broadcast approach} coding scheme over the wireless channel.
\end{abstract}


%
\IEEEpeerreviewmaketitle

\section{Introduction}

{\let\thefootnote\relax\footnote{Part of this work was presented in \cite{dwcnc2017}. 
This work was supported by NSF grant \#$1619129$ and CCF grant \#$1423140$.}}

Internet traffic will soon be dominated by the consumption of what we refer to as \emph{augmented information (AgI) services}.
Unlike traditional information services, in which users consume information that is produced or stored at a given source and is delivered via a communications network, AgI services provide end users with information that results from the real-time {\em processing} of source information via possibly multiple
service functions that can be hosted anywhere in the network. 
Examples 
include real-time analysis of remote sensing data, real-time computer vision, personalized video streaming, and augmented reality, among others.

While today's AgI services are mostly implemented 
in the form of software functions instantiated over general purpose servers at centralized cloud data centers \cite{fxbook}, the increasingly low latency requirements of next generation real-time AgI services is driving cloud resources closer to the end users in the form of small cloud nodes at the edge of the network, resulting in what is referred to as a \emph{distributed cloud network}.
This naturally raises the question of 
where to execute each service function, a question that is impacted both by the computation and the communication resources of the cloud network infrastructure.
The problem of placing an unordered set of service functions in a distributed cloud network was addressed in \cite{vnf}. The authors formulate the problem as a generalization of facility location and generalized assignment, and provide algorithms with bi-criteria approximation guarantees.
The work in \cite{csdp} introduced a network flow model that allows optimizing the distribution (function placement and flow routing) of services with arbitrary function relationships (\eg service chaining) over capacitated cloud networks. Cloud services are described via a directed acyclic graph and the function placement and flow routing is determined by solving a minimum cost network flow problem 
on a cloud-augmented graph.
Ref. \cite{infocom17} provides fast approximation algorithms for the service distribution problem introduced in \cite{csdp},  particularized to service function chains.


The capacity of {\em wireline} cloud networks was recently addressed by the present authors in Refs. \cite{infocom} and \cite{icc}. 
These works provided the first characterization of a cloud network capacity region, in terms of the closure of service input rates 
that can be stabilized by any control algorithm, and the design of throughput-optimal dynamic control policies that achieve an average network cost arbitrarily close to the minimum.

A key aspect not considered in all previous works is the increasingly important role of the wireless access network for efficient service delivery. 
AgI services are increasingly
sourced and accessed from wireless devices, and with the advent of mobile and fog computing \cite{fog}, service functions can also be hosted at wireless computing nodes
(\ie computing devices with wireless networking capabilities) such as mobile handsets, connected vehicles, compute-enabled access points or cloudlets \cite{cloudlet}.
When introducing the wireless network into the computing infrastructure, the often unpredictable nature of the wireless channel further complicates flow scheduling, routing, and resource allocation decisions. 
In the context of traditional wireless communication networks, 
the Lyapunov drift plus penalty (LDP) control methodology (see \cite{Neely_book2} and references therein) has been shown to be a promising approach to tackle these intricate stochastic network optimization problems. 
Ref. \cite{DIVBAR_Neely} extends the LDP approach 
to multi-hop, multi-commodity wireless ad-hoc networks, leading to 
the Diversity Backpressure (DIVBAR) algorithm. DIVBAR exploits the broadcast nature of the wireless medium without precise channel state information (CSI) at the transmitter,
and it is shown to be throughput-optimal under the assumption that at most one packet can be transmitted in each transmission attempt,  and that no advanced coding scheme is used. 
Ref. \cite{DIVBAR_Molisch} extends DIVBAR by incorporating rateless coding in the transmissions of a single packet, yielding enhanced throughput performance.


Motivated by the important role of wireless networks in the delivery of AgI services, in this paper, we address the problem of optimal distribution of AgI services over a multi-hop
{\em wireless computing network}, which is composed of nodes with communication and computing capabilities. 
We extend the {\em multi-commodity-chain} (MCC) flow model of \cite{csdp}, \cite{infocom}, \cite{icc} for the delivery of AgI services over wireless computing networks. 
We adopt the {\em broadcast approach} coding scheme \cite{Shamai_Steiner,Tulino_et_al_2014}, where information is encoded into superposition layers according to the channel conditions.
We 
characterize the capacity region of a wireless computing network and design a fully distributed flow scheduling and resource allocation algorithm that adaptively stabilizes the underlying queuing system while achieving arbitrarily close to minimum network cost, with a tradeoff in network delay.

Our contributions can be summarized as follows:
\begin{enumerate}
\item
We extend the MCC flow model of \cite{csdp}, \cite{infocom}, \cite{icc} for the delivery of AgI services over wireless computing networks, taking into account the routing diversity created by the inherent broadcast nature of the wireless channel. 
In the wireless MCC model, the queue backlog of a given commodity builds up from receiving information units of the same commodity via broadcast transmissions from neighboring nodes, as well as from the generation of information units of the same commodity via local service function processing.
\item We incorporate the use of broadcast approach coding scheme into the scheduling of AgI service flows over wireless computing networks in order to exploit routing diversity and significantly enhance transmission efficiency.
\item For a given set of AgI services, we characterize the capacity region of a wireless computing network in terms of the set of exogenous service input rates 
that can be processed through the required service functions and delivered to the required destinations. Unlike the capacity region of a traditional communications network, which only depends on the network topology, the capacity region of a wireless computing network also depends on the AgI service structure, and it is shown to be enlarged via the use of the broadcast approach.

\item We design a dynamic wireless computing network control (DWCNC) algorithm that makes local transmission, processing, and resource allocation decisions without knowledge of service demands or their statistics, and allow pushing total resource cost arbitrarily close to minimum with a tradeoff in network delay. In particular, DWCNC exhibits a $[O(1/V),O(V)]$ cost-delay tradeoff (where $V$ is a control parameter).
\end{enumerate}


The remainder of the paper is organized as follows: Section II presents the system model. Section III characterizes the network capacity region of a wireless computing network. Section IV
constructs the DWCNC algorithm, and Section V proves the optimal performance of DWCNC. The paper is concluded in Section VI.

\section{System Model}
\label{model}

\subsection{Network Model}

We consider a wireless computing network composed of $N=|\cal N|$ distributed computing nodes 
that communicate over wireless links labeled according to node pairs $(i,j)$ for $i, j \in\cal N$.
Node $i\in\cal N$ is equipped  with $K_{i}^{\text{tr}}$ transmission resource units (\eg transmission power) that it can use to transmit information over the wireless channel. 
In addition, node $i$ is equipped with $K_{i}^{\text{pr}}$ processing resource units (\eg central processing units or CPUs) that it can use to process information as part of an AgI service (see Sec. \ref{service}).

Time is slotted with slots normalized to integer units $t\in\{0,1,2,\dots\}$.
We use the binary variable $y^{\text{tr}}_{i,k}(t)\in\{0,1\}$ to indicate the allocation or activation of $k\in\{0,\dots,K_{i}^{\text{tr}}\}$ transmission resource units at node $i$ at time $t$, which incurs $w^{\text{tr}}_{i,k}$ cost units.
Analogously, $y^{\text{pr}}_{i,k}(t)\in\{0,1\}$ indicates the allocation of $k\in\{0,\dots,K_{i}^{\text{pr}}\}$ processing resource units at node $i$ at time $t$, which incurs $w^{\text{pr}}_{i,k}$ cost units. Notice that the binary resource allocation variables $y^{\text{tr}}_{i,k}(t)$, $y^{pr}_{i,k}(t)$ must satisfy
$\sum\nolimits_{k\in {\mathcal K}_i^{\text{tr}}}{y_{i,k}^{\text{tr}}(t)}\le 1$,
$\sum\nolimits_{k\in {\mathcal K}_i^{\text{pr}}}{y_{i,k}^{\text{pr}}(t)}\le 1$.



\subsection{Augmented Information Service Model}
\label{service}

\begin{figure}
\centering
\includegraphics[width=4in]{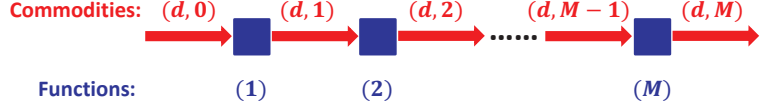}
\vspace{-0.3cm}
\caption{Illustration of an AgI service chain for destination $d\in\cal D$. There are $M$ functions and $M+1$ commodities. The AgI service takes source commodity $(d,0)$ and delivers final commodity $(d,M)$ after going through the sequence of functions $\{1,2,\dots,M\}$.  Function $m$ takes commodity $(d,m-1)$ and generates commodity $(d,m)$. }
\label{service_chain}
\vspace{-0.5cm}
\end{figure}

While the analysis in this paper readily applies to an arbitrary number of services, for ease of exposition, we focus on the distribution of single
 {\em augmented information service}, 
described by a chain of functions $\mathcal M=\{1,2,\dots,M\}$.
A service request is described by a source-destination pair $(s,d)\in\mathcal N \times\mathcal N$,
indicating the request for source flows originating at node $s$ to go through the sequence of functions $\mathcal M$ before exiting the network at destination node $d$.
We adopt a MCC flow model, in which commodity $(d,m)\in \mathcal N \times \{\mathcal M,0\}$ identifies the information units 
generated by function $m\in\mathcal M$ for destination $d\in\mathcal N$. 
We assume information units have arbitrary fine granularity (\eg packets or bits).
Commodity $(d,0)$ denotes the source commodity for destination $d$, which identifies the information units arriving exogenously
at each source node $s$ that have node $d$ as their destination. 
(see Fig. \ref{service_chain}).  



Each service function has (possibly) different processing requirements.
We denote by $r^{(m)}$ the \emph{processing complexity factor} of function $m$, which indicates the number of 
operations required by function $m$ to process one input information unit.
Another key aspect of AgI services is the fact that information flows can change size as they go through service functions.
Let $\xi^{(m)}>0$ denote the {\em scaling factor} of function $m$. Then, the size of the function's output flow is $\xi^{(m)}$ larger than its input flow.

\subsection{Computing Model}

As is shown in Fig. \ref{fig: computing_node}, we represent the processing capabilities of wireless computing nodes via a processing element (\eg CPU in a cloudlet node) co-located with each network node. 
A static, dedicated \emph{computing channel} is considered, where the achievable processing rate at node $i$ with the allocation of $k$ processing resource units is given by $R_{i,k}$ in operations per timeslot.
We use $\mu_{i,\text{pr}}^{(d,m)}(t)$ to denote the flow rate (in information units per timeslot) of commodity $(d,m)$ ($0\le m < M$) from node $i$ to its processing element at time $t$, and $\mu_{\text{pr},i}^{(d,m)}(t)$  to denote the flow rate of commodity $(d,m)$ ($0< m \le M$) from the processing element back to node $i$ (see Fig. \ref{fig: computing_node}).
We then have the following MCC and maximum \emph{processing rate} constraints:
\begin{align}
& \mu_{\text{pr},i}^{(d,m)}(t) = \xi^{(m)} \mu_{i,\text{pr}}^{(d,m\!-\!1)}(t), \quad\,\,\,\,\,\,  \forall i, d,  m\Scale[1]{>}0, t,  \label{chain} \\ 
& \sum\nolimits_{(d,m\Scale[0.6]{>0})} \mu_{i,\text{pr}}^{(d,m\Scale[0.8]{-1})}(t) \ r^{(m)} \le \sum\nolimits_{k=0}^{K_i^{\text{pr}}} R_{i,k} \ y_{i,k}^{\text{pr}}(t), \qquad   \forall i,t.  \label{ratepr}
\end{align}

Note that function $m$ at node $i$ processes input commodity $(d,m-1)$ at a rate $\mu_{i,\text{pr}}^{(d,m-1)}(t)$ information units per timeslot, using $\mu_{i,\text{pr}}^{(d,m-1)}(t) r^{(m)}$ operations per timeslot, and generates output commodity $(d,m)$ at a rate $\mu_{\text{pr},i}^{(d,m)}(t)$ information units per timeslot.

\begin{figure}
\centering
\includegraphics[width=3.5in]{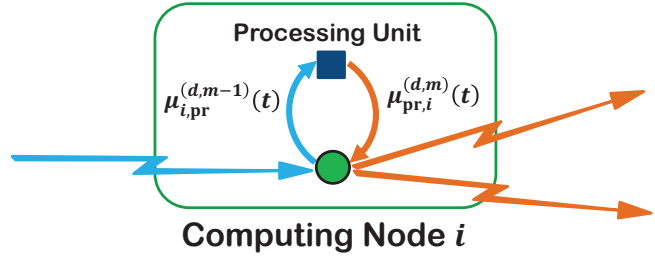}
\vspace{-0.3cm}
\caption{A computing node hosting function $m$ that processes commodity $(d,m-1)$ into commodity $(d,m)$.}
\label{fig: computing_node}
\vspace{-0.5cm}
\end{figure}



\subsection{Wireless Transmission Model}
\label{sec: wireless_transmission_model}


We assume that multiple transmitters (TXs) may transmit simultaneously to overlapping receivers (RXs)
via the use of orthogonal broadcast channels of fixed bandwidth, a priori allocated by a given policy, whose design is outside the scope of this paper.
On the other hand,
due to the broadcast nature of the wireless medium, multiple RXs may overhear the transmission of a given TX.
We model the channel between node $i$ and all other nodes in the network as a  physically degraded Gaussian broadcast channel, 
where the \emph{network state process} (the vector of all channel gains), denoted by ${\bf S}(t) \triangleq \{ s_{ij}(t), \forall i,j \in\cal N\}$, evolves according to 
a Markov process with state space ${\mathcal S}$ and whose steady-state probability exists.
We assume that the statistical CSI is known at the TX,
while the instantaneous CSI can only be learned after the transmission has taken place and is thereby outdated (delayed).

It is well-known that superposition coding is optimal (capacity achieving) for the physically degraded
broadcast channel with independent messages \cite{el2011network}.
In particular,  in this work we adopt the
broadcast approach coding scheme (see \cite{Shamai_Steiner,Tulino_et_al_2014} and references therein),
which consists of sending incremental information using superposition layers, such that
the number of decoded layers at any RX depends on its own channel state,
and the information decoded by a given RX is a subset of the information decoded by any other RX with no worse channel gain.
That is, for a given transmitting node $i$, if we sort the $N-1$ potential receiving nodes in non-decreasing order of their channel gains $\{q_{i,1},\dots,q_{i,N-1}\}$, such that $q_{i,n}$ with $n\in\{1,\dots, N-1\}$ denotes the receiver with the $n$-th lowest channel gain, 
then the information decoded by receiver $q_{i,n}$ is also decoded by receiver $q_{i,u}$, for $ u>n$.
Moreover, let $\Omega_{i,n} \triangleq \{q_{i,n},\cdots,q_{i,N-1}\}$ be the set of receivers with the $N - n$ highest channel gains.
Then, we can partition the information transmitted by node $i$ during a given timeslot into $N-1$ disjoint groups, with the $n$-th partition being the information whose successful receiver set is exactly $\Omega_{i,n}$, \ie the information that is decoded by the nodes in $\Omega_{i,n}$, but not by the nodes in ${\mathcal N} \backslash \{i\} \backslash \{\Omega_{i,n}\}$.
Fig. \ref{fig: broadcast_approach} illustrates the use of the broadcast approach for multi-receiver diversity.


\begin{figure}
\centering
\includegraphics[width=3.5in]{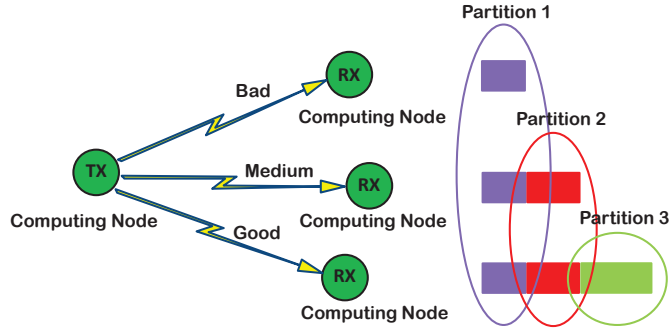}
\vspace{-0.3cm}
\caption{Illustration of the use of the broadcast approach to leverage multi-receiver diversity. The information decoded by the receiver with the ``bad'' channel is a subset of the information decoded by the receiver with the ``medium'' channel, which is further a subset of the information decoded by the receiver with the ``good'' channel. The transmitted information can therefore be grouped into three partitions.}
\label{fig: broadcast_approach}
\vspace{-0.5cm}
\end{figure}

Let  
${p}_{i,k}(a)$ denote the 
optimal power density function 
over the continuum of superposition layers resulting from the allocation
of $k$ transmission resource units at node $i$. 
Then, based on  the broadcast approach \cite{Tulino_et_al_2014}, when allocating $k$ transmission resource units, the maximum achievable rate over link $(i,j)$ at time $t$ is given by
\begin{equation} \label{Rs}
R_{ij,k}(t) = 
 \int_0^{g_{ij}(t)}  \frac{  a {p}_{i,k}(a)}{1 + a \int_a^\infty {p}_{i,k}(s) ds} da,
\end{equation}
where $g_{ij}(t)$ is the channel gain over link $(i,j)$ at time $t$.

In practice, the continuum of superposition layers are discretized into a set of $L_i, \forall i$, discrete code layers. 
In this case, the channel gain of each outgoing link from node $i$ at time $t$ can be discretized into $L_i+1$ states, denoted by $\mathcal S_{i}\triangleq\{\bar s_{i,0},\cdots,\bar s_{i,L_i}\}$, and $L_i$ channel gain thresholds $\{
\bar g_{i,l}, 1\le l \le L_i: \bar g_{i,1}\le\cdots\le \bar g_{i,L_i}\}$. Then, we have
\begin{equation}
{s_{ij}}\left( t \right) =
\begin{cases}
{\bar s_{i,0}},&\text{if ${g_{ij}}(t) < {\bar g_{i,1}};$}\\
{\bar s_{i,l}},&\text{if ${\bar g_{i,l }} \le {g_{ij}}(t) < {\bar g_{i,l+1}},\ 1 \le l < {L_i}-1;$}\\
{\bar s_{i,{L_i}}},&\text{if ${g_{ij}}(t) \ge {\bar g_{i,{L_i}}}.$}
\end{cases} \notag
\end{equation}
Let $P_{i,k}$ denote the total power associated with the allocation of $k$ transmission resource units at node $i$, and $P_{i,k}(l)$ the power allocated to code layer $l$, with $\sum_{l=1}^{L_i}P_{i,k}(l)=P_{i,k}$.\footnote{We assume that the power allocated to each layer is given and leave its optimization out of the scope the paper.}
We then use $\mathcal R_{i,k}\triangleq \{\bar R_{i,k}^0,\cdots,\bar R_{i,k}^{L_i}\}$ to denote the maximum achievable transmission rates associated with the $L_i+1$ channel states, where
\begin{equation}
\bar R_{i,k}^{l} = \sum\limits_{l' \le l} {\log \left( {1 + \frac{{{ P _{i,k}}\left( {{l'}} \right){\bar g_{i,l'}}}}{{1 + {\bar g_{i,l'}}\sum\nolimits_{l'' > l'} {{ P _{i,k}}\left( {{l''}} \right)} }}} \right)},\ {\rm for}\ 1\le l \le L_i,
\label{eq_max_rate_per_layer}
\end{equation}
and $\bar R_{i,k}^{0}=0$.

Hence, the maximum achievable transmission rate over link $(i,j)$ at time $t$ is given by 
\begin{equation}
R_{ij,k}\left( t \right) =
\begin{cases}
{\bar R_{i,k}}^0,&\text{if ${g_{ij}}(t) < {\bar g_{i,1}};$}\\
{\bar R_{i,k}}^l,&\text{if ${\bar g_{i,l }} \le {g_{ij}}(t) < {\bar g_{i,l+1}},\ 1 \le l < {L_i}-1;$}\\
{\bar R_{i,l}}^{L_i},&\text{if ${g_{ij}}(t) \ge {\bar g_{i,{L_i}}}.$}
\end{cases} \notag
\end{equation}

\subsection{Communication Protocol}
\label{sec: communication protocol}

The communication protocol between each TX-RX pair 
is illustrated in Fig. \ref{fig_timing_diagram}.
At the beginning of each timeslot, TX and RX exchange all necessary control signals, including queue backlog state information (see Sec. \ref{sec: queuing_model}). 
Then, the TX decides how many transmission resource units to allocate for the given timeslot and how much rate to allocate to each available commodity. 
Afterwards, the transmission starts and lasts for a fixed time period (within the timeslot); during that time, both data and pilot tones (whose overhead is neglected) are transmitted.

After the transmission ends, every potential RX provides immediate feedback, containing the identification of the information decoded by the RX, which allows the TX to derive the experienced CSI.
The TX then 
makes a \emph{forwarding decision} and sends it through a final instruction to all the RXs,
instructing each RX which portion of its decoded information to keep for further processing and/or forwarding (hence assigning the processing/forwarding responsibility).
Control information, feedbacks, and final instructions are sent through a stable control channel, whose overhead is neglected.

\begin{figure}
\centering
\includegraphics[width=3.5in]{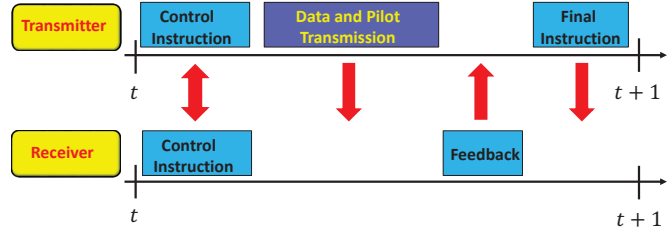}
\vspace{-0.3cm}
\caption{Timing diagram of the communication protocol over a wireless link.}
\label{fig_timing_diagram}
\vspace{-0.5cm}
\end{figure}

We use  $\mu_{ij}^{(d,m)}(t)$ to denote the amount of information of commodity $(d,m)$ retained by node $j$ after the transmission from node $i$ during timeslot $t$.
In addition, it shall be useful to denote by $\mu_{iq_{i,u},n}^{(d,m)}(t)$ the information retained by node $q_{i,u}$ belonging to the $n$-th partition of node $i$'s
transmitted information. 
Then, since $q_{i,u}\in \Omega_{i,n}$ for all $n$ satisfying $n\le u$, we have
\begin{align}
\label{eq_rate_vs_group_rate2}
&\mu _{iq_{i,u}}^{\left( {d,m} \right)} \! \left(t \right) = \sum\nolimits_{n = 1}^{u} \mu^{(d,m)}_{iq_{i,u},n}(t),\quad \forall i, u, d,m,t.
\end{align}

Moreover, according to the broadcast approach, the maximum achievable rate of the $n$-th partition, given the allocation of $k$ transmission resource units at time $t$, is $ R_{i{q_{i,n}},k} (t) - R_{i{q_{i,n-1}},k} (t)$. We then have,
\begin{align}
\label{eq_group_rate}
&\sum\nolimits_{(d,m)} \!  \mu _{i,q_{i,u},n}^{\left( {d,m} \right)} \! \left(t \right)  \le \sum\nolimits_{k=0}^{K_i^{\text{tr}}}
\left[ R_{i{q_{i,n}},k} (t) \!-\! R_{i{q_{i,n-1}},k} (t) \right] y_{i,k}^{\text{tr}}(t),  \ \ \forall i,t,u\ge n, 
\end{align}
where $R_{iq_{i,0},k} \! \left( t \right) = 0$, for all $i,k,t$.

Note that Eqs. \eqref{eq_rate_vs_group_rate2} and \eqref{eq_group_rate} lead to the following rate constraint on link $(i,j)$ for all $t$: $\sum\nolimits_{(d,m)}\mu _{ij}^{\left( {d,m} \right)} \! \left(t \right) \leq \sum\nolimits_{k=0}^{K_i^{\text{tr}}} R_{ij,k}(t) \ y_{i,k}^{\text{tr}}(t)$.



\subsection{Queuing Model}
\label{sec: queuing_model}

We denote by $a_i^{(d,m)}(t)$ the exogenous arrival rate of commodity $(d,m)$ at node $i$ at time $t$,
and by $\lambda_i^{(d,m)}$ its expected value.
We assume that $a_i^{(d,m)}(t)$ is independently and identically distributed (i.i.d.) across timeslots and its forth moment is upper bounded,\footnote{The upper bound of the fourth moment is used in the proof of convergence with probability $1$ in Theorem \ref{thm: network_capacity_region2}.}
\ie $\mathbb{E}\{(\sum_{(d,m)}{a_i^{(d,m)}(t)})^4\}\le A_{\max}^4$. 
Recall that in an AgI service only the source commodity $(d,0)$ enters the network exogenously, while all other commodities are created inside the network as the output of a service function.
Hence, $a_i^{(d,m)}(t)=0$, for all $i,t$ when $m>0$. 

During AgI service delivery, internal network queues buffer incoming data according to their commodities.
We define the \emph{queue backlog} of commodity $(d,m)$ at node $i$, $Q_i^{(d,m)}(t)$, as the amount (in information units) of commodity $(d,m)$ in the queue of node $i$ at the beginning of timeslot $t$, which evolves over time as follows:
\begin{equation}
 Q_i^{(d,m)}\!(t\!+\!1) \leq \left[ Q_i^{(d,m)}\!(t) - {\displaystyle\sum_{j:j\neq i}} \ \mu_{ij}^{\!(d,m)}\!(t) - \mu_{i,\text{pr}}^{\!(d,m)}\!(t) \right]^{\!+}+{\displaystyle\sum_{j:j\neq i}} \ \mu_{ji}^{\!(d,m)}\!(t) + \mu_{\text{pr},i}^{\!(d,m)}\!(t) + a_i^{\!(d,m)}\!(t).
\label{eq_queueing_dynamic}
\end{equation}

Note that in an AgI service only the final commodity $(d,M)$ is allowed to exit the network once it arrives to its destination $d\in\cal D$, while any other commodity $(d,m)$, $m<M$, can only get consumed by being processed into the next commodity $(d,m+1)$ of the service chain. Final commodity $(d,M)$ is assumed to leave the network immediately upon arrival/decoding at its destination, \ie $Q_d^{(d,M)}(t)=0$, for all $d,t$.

\subsection{Network Objective}
\label{sec: network objective}

The goal is to design a control algorithm that dynamically schedules, routes, and process service flows over the wireless computing network with minimum total average resource cost,
\begin{equation}
\underset{t\rightarrow\infty}{\limsup} \quad  \frac{1}{t} \, \sum\nolimits_{\tau=0}^{t-1} \, \mathbb{E}\left\{h(\tau)\right\}, \label{obj2}
\end{equation}
where $h(t)$ is the total cost of the network at time $t$,
\begin{equation}
h\left( t  \right) \triangleq \sum\nolimits_{i\in\cal N} {\left[ {\sum\nolimits_{k=0}^{K_i^{\text{pr}}} {w_{i,k}^{\text{pr}} \ y_{i,k}^{\text{pr}} \! \left( t  \right)}  + \sum\nolimits_{k=0}^{K_i^{\text{tr}}} {w_{i,k}^{\text{tr}} \ y_{i,k}^{\text{tr}} \! \left( t  \right)} } \right]}, \\
\end{equation}
while ensuring that the network is \emph{rate stable} \cite{Neely_book2}, i.e.,
\begin{equation}
\lim_{t\rightarrow\infty} \ \frac{1}{t} \ Q_i^{(d,m)}(t) = 0 \quad \text{with prob. 1,} \ \ \ \forall i, d,m.
\end{equation}

\section{Wireless Computing Network Capacity Region}
\label{sec: network capacity region}

For a given set of AgI services, the wireless computing network capacity region $\Lambda$ is defined as the closure of all service input rate matrices $\{\lambda_i^{(d,m)}\}$ that can be stabilized by a control algorithm. 


\begin{thm}\label{thm: network_capacity_region}
The wireless computing network capacity region $\Lambda$ consists of all average exogenous input rates $\{\lambda_i^{(d,m)}\}$ for which
there exist multi-commodity flow variables $f_{ij}^{(d,m)}$, $f_{\text{\emph{pr}},i}^{(d,m)}$, $f_{i,\text{\emph{pr}}}^{(d,m)}$,
together with 
probability values $\alpha_{i,k}^{\text{\emph{pr}}}$, $\alpha_{i,k}^{\text{\emph{tr}}}(\bf s)$, $\beta_{i,\text{\emph{pr}}}^{(d,m)}(k)$, $\beta_{i,\text{\emph{tr}}}^{(d,m)}({\bf s},k)$, $\eta_{ij}^{(d,m)}({\bf s},k,n)$,
for all $i,j\neq i,k,d,m$, and all network states $\bf{s}\in\cal S$,
such that:
\begin{align}
& \sum\nolimits_{j } \!{f_{ji}^{\left( {d,m} \right)}} \!\! + f_{\text{\emph{pr}},i}^{\left( {d,m} \right)} \!\! + \lambda_i^{(d,m)}  \le
\sum\nolimits_{j} \! {f_{ij}^{\left( {d,m} \right)}}  \!\! + f_{i,\text{\emph{pr}}}^{\left( {d,m} \right)}, \quad \forall i,d,m<M {\rm \ \ or \ \ } \forall i\neq d,m=M  \label{eq_thm1_stability} \\
& f_{\text{\emph{pr}},i}^{\left( {d,m+1} \right)}=\xi^{(m+1)}f_{i,\text{\emph{pr}}}^{\left( {d,m} \right)},\quad \forall i,d, m<M \label{eq_thm1_processing_conservation}\\
& f_{i,\text{\emph{pr}}}^{\left( {d,m} \right)} \le \frac{1}{r^{(m+1)}} \sum\nolimits_{k =0}^{K_{i}^{\text{\emph{pr}}}} {{\alpha _{i,k}^{\text{\emph{pr}}}} \beta _{i,\text{\emph{pr}}}^{\left( {d,m} \right)}(k) {R_{i,k}}},  \quad \forall i,d, m<M
\label{eq_thm1_processing_capacity_constraint} \\
& f_{ij}^{\left( {d,m} \right)} \le \sum\nolimits_{{\bf{s}} \in \cal S} {{\pi _{\bf{s}}}\sum\nolimits_{k=0}^{K_i^{\text{\emph{tr}}}} {\alpha _{i,k}^{\text{\emph{tr}}}({\bf s})\beta _{i,tr}^{\left( {d,m} \right)}
\!\!\left( {{\bf{s}},k} \right)} } \sum\nolimits_{n = 1}^{q_{i,{\bf s}}^{ - 1}\left(j \right)} \! {\left[ {{R_{i{q_{in}},k}} \! \left( {\bf{s}} \right) \! - \! {R_{i{q_{in - 1}},k}} \! \left( {\bf{s}} \right)} \right]\eta _{ij}^{\left( {d,m} \right)} \!\! \left( {{\bf{s}},k,n} \right)} ,\notag\\
&\qquad\qquad\qquad\qquad\qquad\qquad\qquad\qquad\qquad\qquad\qquad\qquad\qquad\qquad\qquad\ \forall i,j,d, m,\label{eq_thm1_flow_capacity_constraint}\\
& f_{i,\text{\emph{pr}}}^{\left( {d,{M }} \right)} = 0,\ f_{\text{\emph{pr}},i}^{\left( {d,0} \right)} = 0,\quad f_{dj}^{\left( {d,{M }} \right)} = 0, 
\ f_{i,\text{\emph{pr}}}^{\left( {d,m} \right)} \geq 0,\ f_{ij}^{\left( {d,m} \right)} \geq 0, \quad \forall i,j,d, m, \label{eq_thm1_flow_positve_conditions} \\
& \sum\nolimits_{k =0}^{{ K}_{i}^{\text{\emph{pr}}}} {{\alpha _{i,k}^{\text{\emph{pr}}}}} \le 1,\ \  \sum\nolimits_{k=0}^{{ K}_{i}^{\text{\emph{tr}}}} {{\alpha _{i,k}^{\text{\emph{tr}}}}}({\bf{s}}) \le 1,\quad  \forall i,\bf s,\label{alphas}  \\
& \sum\nolimits_{\left( {d,m} \right)} {\beta _{i,\text{\emph{pr}}}^{\left( {d,m} \right)}(k)}  \le 1,\ \sum\nolimits_{\left( {d,m} \right)} {\beta _{i,\text{\emph{tr}}}^{\left( {d,m} \right)}({\bf{s}},k)}  \le 1, \quad \forall i,{\bf s },k \label{betas}\\
& \sum\nolimits_{j } {\eta _{ij}^{\left( {d,m} \right)} \!\! \left( {\bf{s}},k,n \right)}  \le 1, \quad \forall i,{\bf s},k,n\label{etas}
\end{align}
where $\bf s$ denotes the network state, whose $(i,j)$-th element $({\bf{s}})_{ij}$ indicates the channel state of link $(i,j)$, $\pi_{\bf s}$ denotes the steady state probability distribution of the network state process ${\bf S}(t)$, and $q_{i,{\bf s}}^{-1}(j)$ in (\ref{eq_thm1_flow_capacity_constraint}) is the index of node $j$ in the sequence $\{q_{i,1},\cdots,q_{i,N-1}\}$, given the network state $\bf s$. Finally, with a slight abuse of notation, $R_{i{j},k}\left( {\bf{s}} \right)$ in (\ref{eq_thm1_flow_capacity_constraint}) denotes the maximum achievable rate over link $(i,j)$, given the network state $\bf s$ and the allocation of $k$ transmission resource units.

Furthermore, the minimum average network cost required for network stability is given by
\begin{align}
\label{eq_thm1_minimum_cost}
\bar h^* = \min \underline h
\end{align}
 where
\begin{equation}
\label{eq:underline}
\underline h = \sum\nolimits_{i\in\cal N} {\left( {\sum\nolimits_{k=0}^{K_i^{\text{\emph{pr}}}} {\alpha _{i,k}^{\text{\emph{pr}}} w_{i,k}^{\text{\emph{pr}}}}  + \sum\nolimits_{k=0}^{K_i^{\text{\emph{tr}}}}{w_{i,k}^{\text{\emph{tr}}} \sum\nolimits_{{\bf{s}} \in \cal S}{\pi_{\bf s}\alpha _{i,k}^{\text{\emph{tr}}}({\bf{s}})}} } \right)},
\end{equation}
and the minimization is over all $\alpha_{i,k}^{\text{\emph{pr}}}$, $\alpha_{i,k}^{\text{\emph{tr}}}(\bf s)$, $\beta_{i,\text{\emph{pr}}}^{(d,m)}(k)$,  $\beta_{i,\text{\emph{tr}}}^{(d,m)}({\bf s},k)$, and $\eta_{ij}^{(d,m)}({\bf s},k,n)$ satisfying (\ref{eq_thm1_stability})-(\ref{etas}).\hfill $\square$
\end{thm}
\begin{proof}
See Appendix \ref{appendix:_proof_of_theorem_1}.
\end{proof}

In the above theorem, 
Eq. \eqref{eq_thm1_stability} are flow conservation constraints,\footnote{Note that final commodity $(d,M)$ staisfies flow conservation at all nodes except at its destination $d$, where it is immediately consumed upon arrival.} Eqs. \eqref{eq_thm1_processing_capacity_constraint} and \eqref{eq_thm1_flow_capacity_constraint} are rate constraints, and
Eq. \eqref{eq_thm1_flow_positve_conditions} indicates non-negativity and flow efficiency constraints. 
The probability values $\alpha_{i,k}^{*\text{pr}}$, $\alpha_{i,k}^{*\text{tr}}(\bf s)$, $\beta_{i,\text{pr}}^{*(d,m)}(k)$,  $\beta_{i,\text{tr}}^{*(d,m)}({\bf s},k)$ and $\eta_{ij}^{*(d,m)}({\bf s},k,n)$ define a \emph{stationary randomized policy} that uses \emph{single-copy routing} -- only one copy of each information unit is allowed to flow through the network - and it is optimal  among all stabilizing algorithms (including  algorithms that use \emph{multi-copy routing}). Specifically, the parameters of the stationary randomized policy are defined as:
\begin{itemize}
\item  $\alpha_{i,k}^{\text{pr}}$: the probability that $k$ processing resource units are allocated at node $i$
\item  $\alpha_{i,k}^{\text{tr}}({\bf s})$: the conditional probability that $k$ transmission resource units are allocated at node $i$, given the network state $\bf s$;
\item  $\beta_{i,\text{pr}}^{(d,m)}(k)$: the conditional probability that node $i$ processes commodity $(d,m)$, given the allocation of $k$ processing resource units;
\item  $\beta_{i,\text{tr}}^{(d,m)}({\bf s},k)$: the conditional probability that node $i$ transmits commodity $(d,m)$, given the network state $\bf s$ and the allocation of $k$ transmission resource units;
\item  $\eta_{ij}^{(d,m)}({\bf s},k,n)$: the conditional probability that node $i$ forwards the information of commodity $(d,m)$ in the $n$-th partition to node $j$, when the network state is $\bf s$
and $k$ transmission resource units are allocated. 
\end{itemize}


It is important to note that this optimal stationary randomized policy is hard to compute in practice, as it requires the knowledge of $\{\lambda_i^{(d,m)}\}$ and solving a complex nonlinear program. However, its existence is essential for proving the performance of our proposed algorithm.

\section{Dynamic Wireless Computing Network Control Algorithm}
\label{sec: algorithm}

We propose a dynamic wireless computing network control strategy that accounts for both transmission and processing flow scheduling and resource allocation decisions in a fully distributed manner.

\subsection{The DWCNC Algorithm}
\label{subsec: alg_description}
{\bf Dynamic Wireless Computing Network Control (DWCNC):}

\underline{Local processing decisions:}
At the beginning of timeslot $t$, each node $i$ observes its local queue backlogs and performs the following operations:

\begin{enumerate}

\item Compute the \emph{processing utility weight} of each commodity $(d,m)$:
\begin{align}
W_i^{(d,m)}(t) \! \buildrel \Delta \over = \! \frac{1}{r^{(m+1)}} \! \left[ Q_i^{(d,m)}(t) \!-\! \xi^{(m\Scale[0.7]{+1})}Q_i^{(d,m\Scale[0.7]{+1})}(t)  \right]^{\!+} , \notag
\end{align}
where we denote $[x]^+= \max\{x,0\} $. Note that $W_i^{(d,m)}(t)$ indicates the ``potential benefit''
of executing function $(m\!+\!1)$ to process commodity $(d,m)$ into commodity $(d,m\!+\!1)$ at time $t$, in terms of local congestion reduction per processing operation. 

\item Compute the optimal number of processing resource units $k_{\text{pr}}^\dag$ to allocate and the optimal commodity $(d,m)_{\text{pr}}^{\dag}$ to process:
\begin{equation}
\label{eq_processing_metric}
\left[ {k_{\text{pr}}^\dag,(d,m)_{\text{pr}}^\dag} \right]\! =\! \mathop {\arg \max }\limits_{k,\left( {d,m} \right)} \left\{ {{R_{i,k}}W_i^{(d,m)}\!(t) - Vw_{i,k}^{\text{pr}}} \right\},
\end{equation}
where $V$ is a non-negative control parameter that determines the degree to which cost is emphasized.

\item Make the following flow rate assignment decisions:
\begin{align}
& \mu_{i,\text{pr}}^{(d,m)_{\text{pr}}^\dag}(t) = {{{R_{i,{k_{\text{pr}}^\dag}}}} \mathord{\left/{\vphantom {{{R_{i,{k^\dag}}}} {{r^{(m+1)}}}}} \right.\kern-\nulldelimiterspace} {{r^{(m^\dag_{\text{pr}}+1)}}}}; \notag\\
& \mu_{i,\text{pr}}^{(d,m)}(t) = 0, \quad \forall (d,m)\neq (d,m)_{\text{pr}}^\dag. \notag
\end{align}

\end{enumerate}

\underline{Local wireless transmission decisions:}
At the beginning of timeslot $t$, each node $i$ observes its local queue backlogs, the queue backlogs of its potential RXs and the associated statistical CSI, and performs the following operations:
\begin{enumerate}
\item For each outgoing link $(i,j)$, compute the \emph{differential backlog weight} of each commodity $(d,m)$: 
\begin{equation}
W_{ij}^{(d,m)}(t) \triangleq \left[Q_i^{(d,m)}(t)-Q_j^{(d,m)}(t)\right]^+.\notag
\end{equation}
\item For each transmission resource allocation choice $k\in\{0,\dots, K_i^{\text{tr}}\}$, compute the \emph{transmission utility weight} 
of each commodity $(d,m)$:
\begin{align}
&\Scale[0.95]{W_{i,k,\text{tr}}^{(d,m)}\left( t \right) \triangleq \sum\limits_{{\bf{s}} \in S} {\Pr \left( {\left. {{\bf{S}}(t) = {\bf{s}}} \right|{\bf{S}}\left( {t - 1} \right)} = \tilde{\bf{s}} \right)}  \sum\limits_{n = 1}^{N - 1} \!\! {\left[ {{R_{i{q_{i,n}},k}}\left( {\bf{s}} \right) - {R_{i{q_{i,n - 1}},k}}\left( {\bf{s}} \right)} \right] \! \mathop {\max }\limits_{j \in {\Omega _{i,n}}\left( {\bf{s}} \right)} \!\left\{ {W_{ij}^{(d,m)} \! \left( t \right)} \! \right\}} },
\label{Wp}
\end{align}
where $ \tilde{\bf{s}}$ denotes the CSI feedbacks at time $t-1$,  
and, with an abuse of notation, ${\Omega _{i,n}}({\bf{s}})$ is used to indicate the dependence of ${\Omega _{i,n}}$ on the network state $\bf s$.
\label{step: compute_trans_utility_weight}
\item Compute the optimal number of transmission resource units $k_{\text{tr}}^{\dag}$ to allocate and the optimal commodity $(d,m)_{\text{tr}}^{\dag}$ to transmit:
\begin{equation}
\label{eq_transmission_metric}
\left[ {{k_{\text{tr}}^{\dag}},{{\left( {d,m} \right)}_{\text{tr}}^{\dag}}} \right] = \mathop {\arg \max }\limits_{k,\left( {d,m} \right)} \left\{ {W_{i,k,\text{tr}}^{(d,m)}\left( t \right)}-Vw_{i,k}^{\text{tr}} \right\}. 
\end{equation}
If $k_{\text{tr}}^{\dag}=0$, node $i$ keeps silent in timeslot $t$.
\label{step: choose_commodity_resource}

\item After receiving the CSI feedbacks, node $i$ identifies the information decoded by each RX and assigns the processing/forwarding responsibility for the $n$-th partition of the transmitted information to the RX in $\Omega_{i,n}({\bf S}(t))$ with the largest positive $W_{ij}^{(d,m)}(t)$, while all other RXs in $\Omega_{i,n}({\bf S}(t))$ and node $i$ discard the information. If no receiver in $\Omega_{i,n}({\bf S}(t))$ has positive $W_{ij}^{(d,m)}(t)$, node $i$ retains the information of partition $n$, while all the receivers in $\Omega_{i,n}({\bf S}(t))$ discard it.
\label{step: forwarding_decision}

\end{enumerate}




\bigskip


Remarks:
\begin{itemize}
\item In the local processing decisions,
maximizing the metric in \eqref{eq_processing_metric} over $[k,(d,m)]$ can be decomposed into first maximizing $W_i^{(d,m)}(t)$ over $(d,m)$ and
then maximizing the metric over $k$ given the maximized $W_i^{(d,m)}(t)$. The computational complexity is 
$O(K_{i}^{\text{pr}} + N M)$.

\item In Step \ref{step: compute_trans_utility_weight} of the local transmission decisions, $W_{i,k,\text{tr}}^{(d,m)}(t)$ is computed according to \eqref{Wp} using the transition probabilities $\Pr(\left.{\bf S}(t)={\bf s}\right|{\bf S}(t-1)=\tilde {\bf s})$, known as the statistical CSI, but the complexity can be high due to the possibly exponentially large network state space with respect to the number of links. However, the computation can be significantly simplified when the channel realizations of the wireless links are mutually independent, which is described in the next subsection.

\item The fact of discarding decoded information at the RXs that do not get the processing/forwarding responsibility during Step \ref{step: forwarding_decision} of the local transmission decisions, implies that DWCNC is a single-copy routing algorithm.
\end{itemize}

\subsection{Transmission Utility Weight with Independent Links and Discrete Code Layers}

Recall that, in practice, when using the broadcast approach, each node uses a $L_i$ discrete code layers, with $R_{ij,k}(t)$ taking values in $R_{i,k}$ as described in Sec. \ref{sec: wireless_transmission_model}).

Let $\bar \Omega_{i,l}({\bf S}(t))$ denote the set of receivers 
that have channel gain no smaller than $\bar g_{i,l}$ at time $t$, \ie $g_{ij}(t)\ge \bar g_{i,l}$ for all $j\in \bar \Omega_{i,l}({\bf S}(t))$, and $g_{ij}(t)< \bar g_{i,l}$ for all $j\notin \bar \Omega_{i,l}({\bf S}(t))$.

Given ${\bf S}(t)=\bf s$ and $k$, we have the following two possible cases for the maximum achievable transmission rate of the $n$-th partition: i) ${R_{i{q_{i,n}},k}}\left( {\bf{s}} \right) - {R_{i{q_{i,n - 1}},k}}\left( {\bf{s}} \right)=0$; ii) ${R_{i{q_{i,n}},k}}\left( {\bf{s}} \right) - {R_{i{q_{i,n - 1}},k}}\left( {\bf{s}} \right) = \sum\nolimits_{l = l_0}^{l_1} {\left( {\bar R_{i,k}^l - \bar R_{i,k}^{l - 1}} \right)} $, for some $l_0$ and $l_1$ satisfying $1\le l_0 \le l_1\le L_i$, with ${{\bar \Omega }_{i,l}}({\bf{s}}) = {\Omega _{i,n}}({\bf{s}})$ for $l_0\le l\le l_1$. Then we have, for all $i$, ${\bf s}$, $k$, $(d,m)$, and $t$,
\begin{equation}
\Scale[1]{\sum\limits_{n = 1}^{N - 1}  \left[ {{R_{i{q_{i,n}},k}}\left( {\bf{s}} \right) - {R_{i{q_{i,n - 1}},k}}\left( {\bf{s}} \right)} \right]\mathop {\max }\limits_{j \in {\Omega _{i,n}}\left( {\bf{s}} \right)} \left\{ {W_{ij}^{(d,m)}\left( t \right)} \right\} = \sum\limits_{l = 1}^{{L_i}} {\left( {\bar R_{i,k}^l - \bar R_{i,k}^{l - 1}} \right)\mathop {\max }\limits_{j \in {{\bar \Omega }_{i,l}}\left( {\bf{s}} \right)} \left\{ {W_{ij}^{(d,m)}\left( t \right)} \right\}}},\notag
\end{equation}
based on which we can rewrite Eq. \eqref{Wp} as follows:
\begin{align}
W_{i,k,\text{tr}}^{(d,m)}\left( t \right) &= \sum\nolimits_{l = 1}^{{L_i}} {\left( {\bar R_{i,k}^l - \bar R_{i,k}^{l - 1}} \right)\sum\nolimits_{{\bf{s}} \in \mathcal S} {\Pr \left( {\left. {{\bf{S}}(t) = {\bf{s}}} \right|{\bf{S}}\left( {t - 1} \right) } \right)\mathop {\max }\nolimits_{j \in {{\bar \Omega }_{i,l}}\left( {\bf{s}} \right)} \left\{ {W_{ij}^{(d,m)}\left( t \right)} \right\}} } \notag\\
&=\sum\nolimits_{l = 1}^{{L_i}} {\left( {\bar R_{i,k}^l - \bar R_{i,k}^{l - 1}} \right)\mathbb{E}\left\{ \left. {\mathop {\max }\nolimits_{j \in {{\bar \Omega }_{i,l}}\left( {{\bf{S}}(t)} \right)} \left\{ {W_{ij}^{(d,m)}\left( t \right)} \right\}} \right| {\mathcal H}(t)\right\}},
\label{eq_trans_utility_weight2}
\end{align}
where $\mathcal H(t)\triangleq \{{\bf Q}(t), {\bf S}(t-1)\}$ is the ensemble of queue backlog observations at time $t$ and the CSI feedbacks at time $t-1$.


Let $1_{ij,l}^{(d,m)}(t)$ denote the indicator that takes value $1$ if receiver $j$ has the largest differential backlog $W_{ij}^{(d,m)}(t)$ among the receivers in $\bar \Omega_{i,l}({\bf S}(t))$, and $0$ otherwise. Then, we have
\begin{align}
\mathbb{E}\left\{ {\left. {\mathop {\max }\nolimits_{j \in {{\bar \Omega }_{i,l}}\left( {{\bf{S}}(t)} \right)} \left\{ {W_{ij}^{(d,m)}\left( t \right)} \right\}} \right|\mathcal H(t)} \right\} &= \mathbb{E}\left\{ {\left. {\sum\nolimits_j {W_{ij}^{(d,m)}\left( t \right)1_{ij,l}^{\left( {d,m} \right)}\left( {{\bf{S}}(t)} \right)} } \right|\mathcal H(t)} \right\}\nonumber\\
&= \sum\nolimits_j {W_{ij}^{(d,m)}\left( t \right)\varphi _{ij,l}^{\left( {d,m} \right)}\left( {\mathcal H (t)} \right)},
\label{eq_trans_utility_weight3}
\end{align}
where $\varphi _{ij,l}^{\left( {d,m} \right)}\left( {\mathcal H (t)} \right)$ is the conditional probability that $1_{ij,l}^{(d,m)}(t)$ takes value $1$ given $\mathcal H (t)$.

Plugging \eqref{eq_trans_utility_weight3} into \eqref{eq_trans_utility_weight2} to compute $W_{i,k,\text{tr}}^{(d,m)}\left( t \right)$, we can replace Step \ref{step: compute_trans_utility_weight} of the local transmission decisions of DWCNC in Sec. \ref{subsec: alg_description} with the following two sub-steps:
\begin{enumerate}[2a)]
\item For each commodity $(d,m)$, sort the receivers of node $i$ according to their differential backlog weight $W_{ij}^{(d,m)}(t)$  in non-increasing order. Let $\Psi_{ij}^{(d,m)}(t)$ denote the set of receivers of node $i$ with index smaller
than the index of receiver $j$ in the sorted list at time $t$. In this case, each receiver in $\Psi_{ij}^{(d,m)}(t)$ has no smaller differential backlog weight than receiver $j$.
\item For each transmission resource allocation choice $k\in\{0,\dots, K_i^{\text{tr}}\}$, compute the \emph{transmission utility weight} 
of each commodity $(d,m)$:
\begin{equation}
W_{i,k,\text{tr}}^{(d,m)}\left( t \right) = \sum\nolimits_{l = 1}^{{L_i}} {\left( {\bar R_{i,k}^l - \bar R_{i,k}^{l - 1}} \right)\sum\nolimits_j {W_{ij}^{(d,m)}\left( t \right)\varphi _{ij,l}^{\left( {d,m} \right)}\left( {\mathcal H(t)} \right)} }.
\end{equation}
\end{enumerate}

Computing $\varphi _{ij,l}^{\left( {d,m} \right)}\left( {\mathcal H(t)} \right)$, which requires the statistical CSI, can be significantly simplified if the channel
realizations of the links are mutually independent. In this case, the value of $\varphi _{ij,l}^{\left( {d,m} \right)}\left( {\mathcal H(t)} \right)$ can be obtained from a simple
multiplication as follows:
\begin{align}
&\varphi _{ij,l}^{\left( {d,m} \right)}\left( {\mathcal H(t)} \right) = \Pr \left( {\left. {{g_{ij}}\left( t \right) \ge {\bar g_{i,l}}} \right|{s_{ij}}\left( {t - 1} \right)} \right)\prod\nolimits_{v \in \Psi _{ij}^{(d,m)}(t)} {\Pr \left( {\left. {{g_{iv}}\left( t \right) < {\bar g_{i,l}}} \right|{s_{iv}}\left( {t - 1} \right)} \right)}\notag \\
&= \sum\nolimits_{l' = l}^{{L_i}} {\Pr \left( {\left. {{s_{ij}}\left( t \right) = {\bar s_{i,l'}}} \right|{s_{ij}}\left( {t - 1} \right)} \right)} \prod\nolimits_{v \in \Psi _{ij}^{(d,m)}(t)} {\sum\nolimits_{l' = 0}^{l - 1} {\Pr \left( {\left. {{s_{iv}}\left( t \right) = {\bar s_{i,l'}}} \right|{s_{iv}}\left( {t - 1} \right)} \right)} },
\label{eq_trans_utility_weight4}
\end{align}
where $\Pr(\left.s_{ij}(t)=\bar s_{i,l'}\right|s_{ij}(t-1)=\bar s_{i,l''})$ is the statistical CSI of link $(i,j)$.

The computational complexity associated with the transmission decisions made by node $i$ at each timeslot is $O(MN(\log_2N+L_iK_i^{\text{tr}}))$, which is dominated by sorting the receivers according to their differential backlogs weights, for each commodity, and computing the transmission utility weights for all commodities and resource allocation choices.

\section{Performance Analysis}
\label{sec: optimality}

In this section, we analyze the throughput-optimality and average cost performance of DWCNC, described by the following theorem.

\begin{thm}
\label{thm: network_capacity_region2}
For any exogenous input rate matrix $\bm \lambda \buildrel \Delta \over = \{\lambda_i^{(d,m)}\}$ strictly interior to the capacity region $\Lambda$,
DWCNC stabilizes the wireless computing network, while achieving an average total resource cost arbitrarily close to the minimum average cost $\overline h^*({\bm \lambda})$
with probability $1$; \ie 
\begin{align}
& \limsup_{t \rightarrow \infty} \frac{1}{t}  \sum\nolimits_{\tau=0}^{t-1}  h(\tau) \leq \overline h^*({\bm \lambda}) +\frac{NB}{V}, \quad \text{with prob. 1,} \label{que}  \\
& \underset{t \rightarrow \infty}{\limsup} \frac{1}{t}   \displaystyle \sum\nolimits_{\tau, i,d,m}  Q_i^{(d,m)}(\tau) \leq \frac{NB + V\left[ \overline h^*
( {\bm \lambda} + \epsilon {\bf 1}  ) - \overline h^* ({\bm \lambda}) \right] }{\epsilon}, \quad \text{with prob. 1,}\label{eq_ave_cost}
\end{align}
where $B$ is a constant that depends on the system parameters $R_{i,K_i^{\text{\emph{tr}}}}(\bf s)$, $R_{i,K_i^{\text{\emph{pr}}}}$, $A_{\max}$,
$\xi^{(m)}$, 
and $r^{(m)}$; $\epsilon$ is a positive constant satisfying $\left({\bm \lambda}+\epsilon \bf{1}\right)\in \Lambda$; 
and $\overline h^*({\bm\lambda})$ denotes the average cost obtained by the optimal solution given input rates $\bm \lambda$.\hfill $\square$
\end{thm}


\begin{proof}
See Appendix \ref{appendix: proof of theorem 2}.
\end{proof}

The finite bound on the expected total queue length in Theorem \ref{thm: network_capacity_region2} implies that the wireless computing network is strongly stable. The parameter $V$ can be increased to
push the average resource cost arbitrarily close to the minimum cost required for network stability, $\overline h^*({\bm \lambda})$, with a linear increase in average network congestion or queueing delay. Thus, Theorem \ref{thm: network_capacity_region2} demonstrates a $[O(1/V),O(V)]$ cost-delay tradeoff.

\section{Numerical Experiments}

In this section, we present numerical results obtained from simulating the performance of the DWCNC algorithm for the delivery of two AgI services over an eleven node network during $10^6$ timeslots. The numerical values presented in this section for resource allocation costs, communication flow rates, processing flow rates, and queue backlogs are all measured in normalized units.

\subsection{Network Structure}

We consider the wireless computing network illustrated  in Fig. \ref{fig: simulation_network}. All $11$ nodes represent computing locations: nodes $1$, $6$, and $7$ represent access points (APs), while all other nodes are user-end (UE) devices. We list the $(X,Y)$ coordinations of all the nodes' locations in Table \ref{tab: node coordinations} in normalized distance units. 
In terms of processing resources, each AP has $5$ resource allocation choices, with associated cost and processing rate $w_{i,k}^{\text{pr}} = k$ and $R_{i,k}=20 k$ for $i\in\{1, 6, 7\}$ and $k\in\{0,1,\dots,4\}$,  while each UE has $2$ resource allocation choices, with associated cost and rate $w_{i,k}^{\text{pr}} = 2k$ and $R_{i,k}=20k$ for $i\in\{2,3,4,5,8,9,10,11\}$ and $k\in\{0,1\}$. Note that processing cost is lower at the APs than at the UEs. 
In terms of transmission resources, each node has $2$ resource allocation choices of cost $w^{\text{tr}}_{i,0}=0$ and $w^{\text{tr}}_{i,1}=1$. 
The associated transmission rates are given in Section \ref{comm_setup}.


The edges in Fig. \ref{fig: simulation_network} represent the active wireless links, whose channel realizations are mutually independent. In addition, the channel realizations of each link are independently and identically distributed (i.i.d.) across timeslots.\footnote{Note that i.i.d. channel state evolution 
is approximately fulfilled when the timeslot length is equal to the coherence time of the wireless medium.} We assume there is Line of Sight (LOS) between APs, and the links between APs have Rician fading (see Ref. \cite{molisch_book}) with the Rice factor equal to $5$ dB. On the other hand, we assume the rest of the links exhibit Rayleigh fading (see Ref. \cite{molisch_book}), where no LOS exists. The path loss coefficient of each link is $3$.

\begin{figure}
\centering
\includegraphics[width=4.2in]{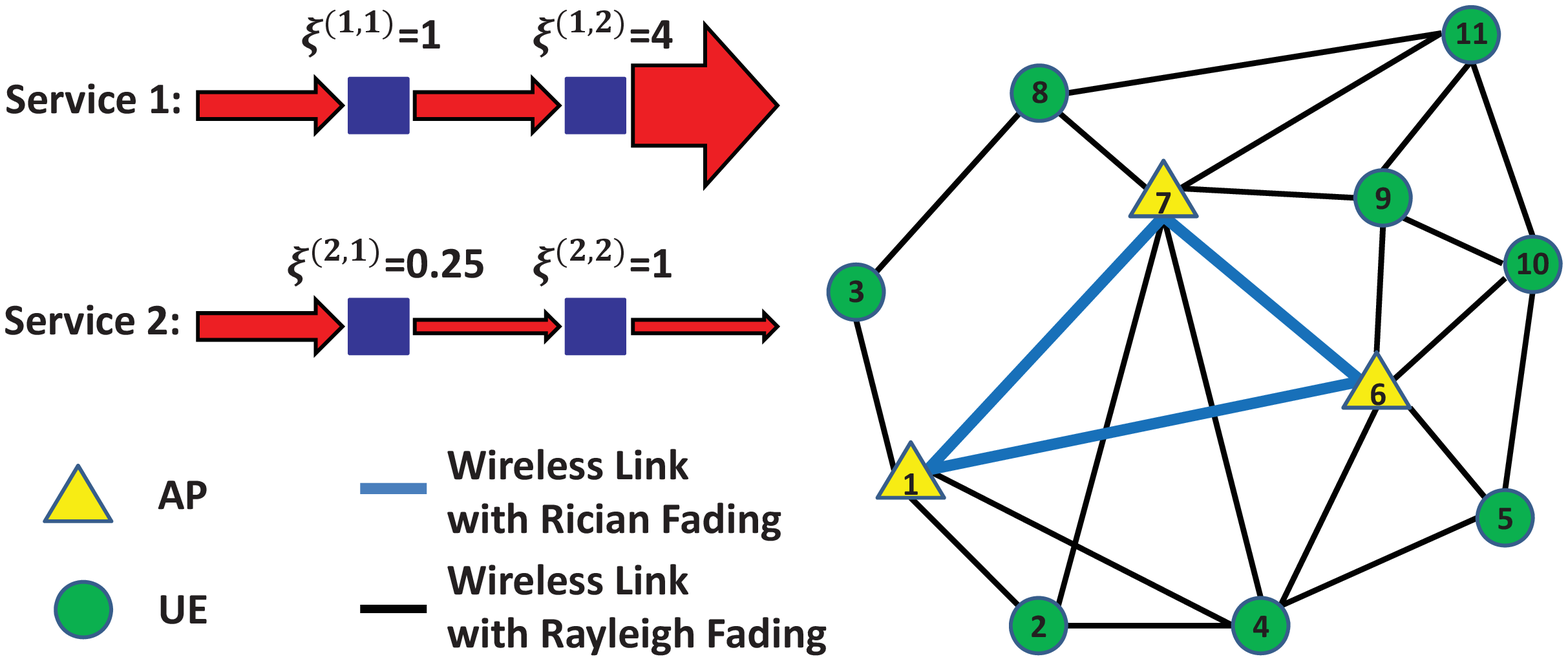}
\caption{A wireless computing network with 3 access points and 8 end user devices, providing 2 AgI services.}
\label{fig: simulation_network}
\end{figure}

\begin{table}
\centering
\caption{Computing Nodes' Locations}
\begin{tabular}{l| l l l l l l l l l l l}
  \hline
  Node Index &1&2&3&4&5&6 \\ \hline
  Location $(X,Y)$  & (0,10)& (10,0) & (-5, 20) & (22, 0) & (27, 5) & (24, 10) \\
  \hline\hline
  Node Index &7&8&9&10&11\\ \hline
  Location $(X,Y)$ & (13, 22) & (5, 30) & (27, 23) & (35, 21) & (30, 33)\\
  \hline
\end{tabular}
\label{tab: node coordinations}
\vspace{-0.5cm}
\end{table}

\subsection{Communication Setup}
\label{comm_setup}

We assume that when using the broadcast approach, nodes use three code layers with corresponding channel gain thresholds $[\bar g_{i,1}=-46.82, \bar g_{i,2}=-40.80, \bar g_{i,3}=-39.03]$ dB for Rician fading channels, and $[\bar g_{i,1}=-43.02, \bar g_{i,2}=-38.37, \bar g_{i,3}=-36.41]$ dB for Rayleigh fading channels.
Allocating the total power among the code layers in the ratio $[1:4:6]$ and $[1:2.9:4.6]$ of the total transmission power, 
respectively for Rician and Rayleigh channels,\footnote{The optimization of the power allocation among different code layers at each transmitting node is beyond the scope of this paper. In this paper, we treat the power values allocated among code layers as given parameters.} and using Eq. \eqref{eq_max_rate_per_layer}, we generate the maximum achievable transmission rates
$[\bar R_{i,1}^0=0, \bar R_{i,1}^1=12.1, \bar R_{i,1}^2=20.6, \bar R_{i,1}^3=48.3]$ for Rician  channels, and
$[\bar R_{i,1}^0=0, \bar R_{i,1}^1=7.8, \bar R_{i,1}^2=13.1, \bar R_{i,1}^3=27.8]$ for Rayleigh  channels.

To demonstrate the efficiency of adopting the broadcast approach, we also simulate the case of adopting the traditional \emph{outage approach} coding scheme. 
With the outage approach, the transmission rate is fixed, and the information is reliably decoded when the instantaneous channel gain exceeds a threshold, otherwise no information is decoded. For each node $i$ using the outage approach, we set $\bar g_i^{\text{out}}=-40.80$ dB and $-38.37$ dB as the outage thresholds 
for Rician 
and Rayleigh outgoing links, respectively, 
such that if $g_{ij}(t)\ge \bar g_i^{\text{out}}$ and $k=1$, the maximum achievable rate $R_{ij,1}(t)$ is equal to the outage rate denoted by $\bar R_{i,1}^{\text{out}}$; otherwise, $R_{ij,1}(t)$ is zero. Note that the values of $\bar g_{i}^{\text{out}}$ and $\bar R_{i,1}^{\text{out}}$ are respectively equal to the values of $\bar g_{i,2}$ and $\bar R_{i,1}^2$ in the broadcast approach.

\subsection{Service Setup}

The wireless computing network offers two services (see Fig. \ref{fig: simulation_network}), each of which consists of two functions. To indicate multiple services, we let $(\phi,m)$, $\phi=1,2$, $m=1,2$, denote the $m$-th function of service $\phi$; and $(d,\phi,m)$, $d\in \{1,\cdots,11\}$, $\phi =1,2$, $m=0,1$ denote the commodity processed by function $(\phi,m+1)$ for destination $d$.

All four functions have the same complexity factor, equal to $1$ (number of operations per unit flow). In terms of flow scaling, as is shown in Fig. \ref{fig: simulation_network}, functions $(1,1)$ and $(1,2)$ have scaling factors $1$ and $4$, respectively, \ie $\xi^{(1,1)}=1$ and $\xi^{(1,2)}=4$; and functions $(2,1)$ and $(2,2)$ have scaling factors $0.25$ and $1$, respectively, \ie $\xi^{(2,1)}=0.25$ and $\xi^{(1,2)}=1$. Note that function $(1,2)$ is an expansion function, while function $(2,1)$ is a compression function.

\subsection{Broadcast Approach v.s. Outage Approach}

We consider a scenario in which each service is requested by $56$ clients corresponding to all possible UE source-destination pairs.

\begin{figure*}[ht]
\centering
\subfigure[]{
\centering \includegraphics[width=3.1in]{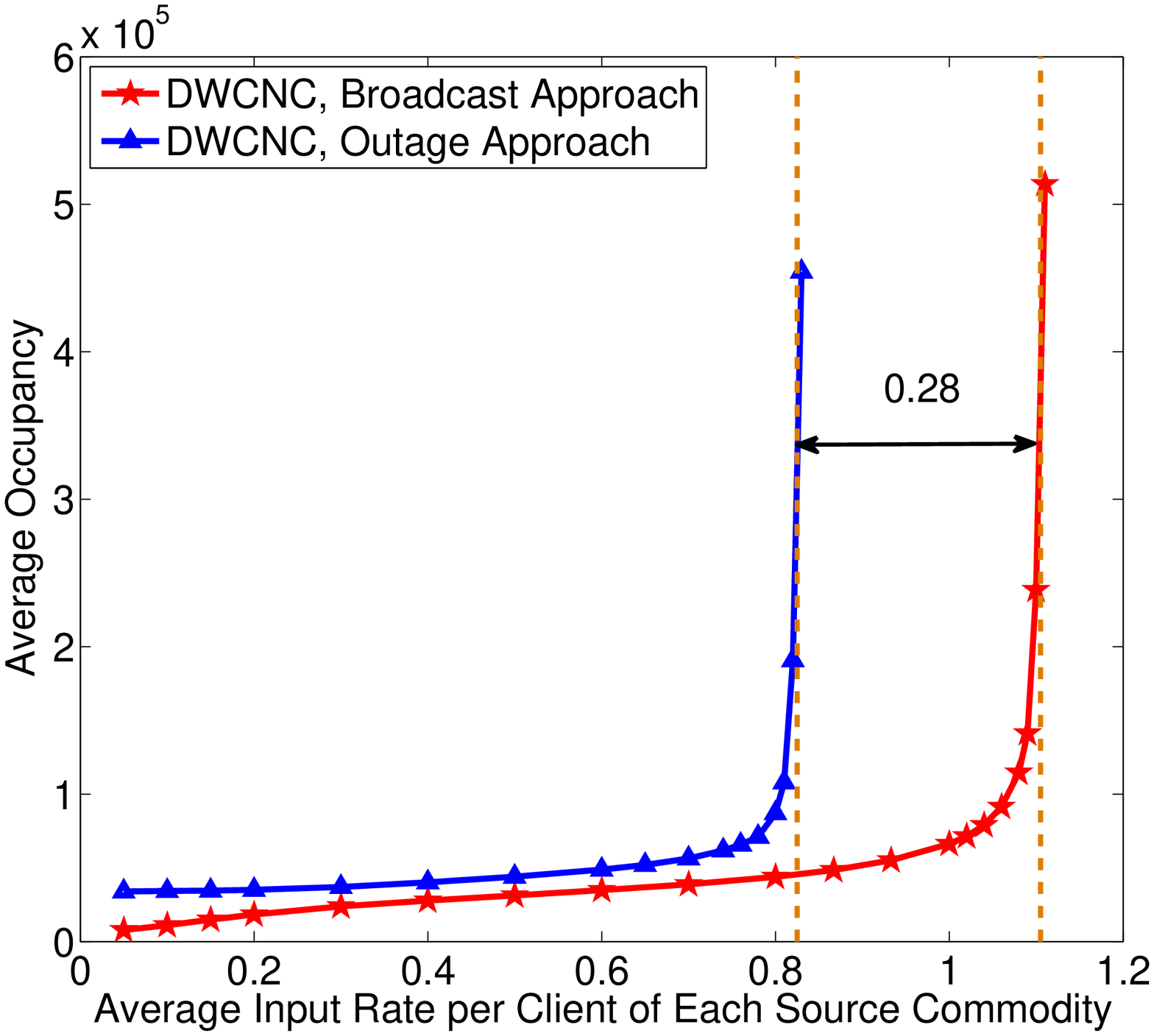}
\label{result3_1}
}
\subfigure[]{
\centering \includegraphics[width=3.1in]{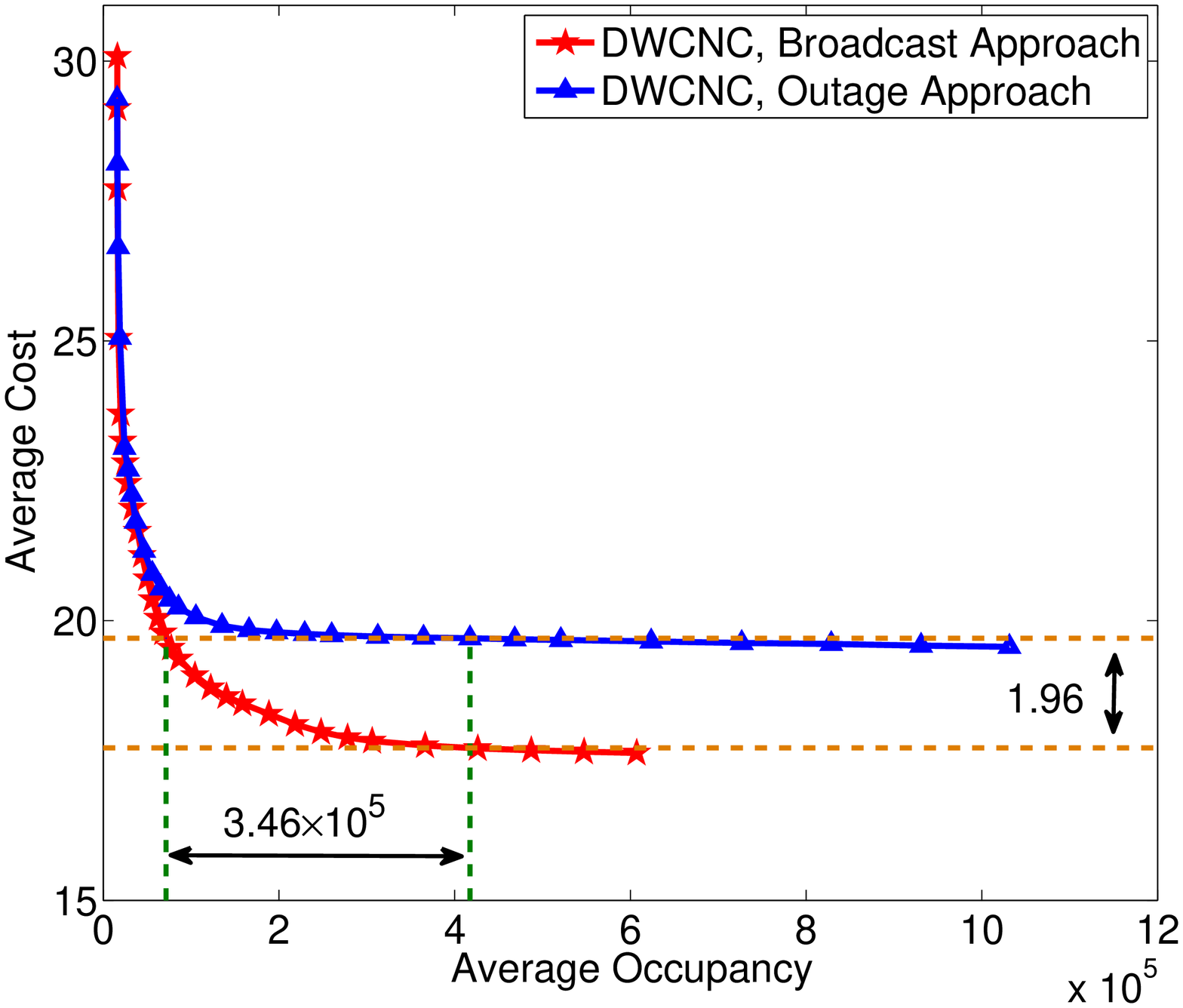}
\label{result3_2}
}

\caption{Performance of DWCNC with the broadcast approach and the outage approach in the large scale scenario: a) Average Cost vs. Average Occupancy  b) Average occupancies evolving with varying exogenous input rate: throughput optimality }
\vspace{-0.5cm}
\label{fig: result3}
\end{figure*}

The throughput performance of DWCNC with both the broadcast approach and the outage approach is shown in Fig. \ref{result3_1}, where we plot the time average occupancy as a function of the  average exogenous input rate (assumed to be the same for all source commodities), while setting the control parameter $V$ equal to $500$.
Observe how the average occupancy exhibits a sharp increase when the exogenous input rate reaches approximately $0.83$ and $1.11$, for the outage and broadcast approach, respectively. According to Theorem \ref{thm: network_capacity_region2}, and considering $\epsilon\rightarrow 0$, this sharp increase indicates that the average input rate has reached the boundary of the computing network capacity region, and hence it is indicative of the maximum achievable throughput. It can be seen from Fig. \ref{result3_1} that the maximum throughput using the broadcast approach is larger than that of using the the outage approach. This significant throughput difference is a clear indication the enhanced transmission ability of the broadcast approach.

In the following, we assume an average input rate of $0.7$, which is interior to the capacity region of DWCNC with both the outage and the broadcast approach.

Fig. \ref{result3_2} shows the tradeoff between the average cost and the average occupancy (total queue backlog) as the control parameter $V$ varies between $0$ and $10^4$, when running DWCNC with the broadcast approach and the outage approach.
It can be seen from Fig. \ref{result3_2} that, with either coding scheme, the average cost decreases with the increase of the average occupancy. In general, both evolutions follow the $[O(1/V),O(V)]$ cost-delay tradeoff of Theorem \ref{thm: network_capacity_region2}. However, the corresponding trade-off ratios are different.
Comparing the two curves in Fig. \ref{result3_2}, it can be seen that the broadcast approach exhibits a better cost-delay tradeoff, in the sense that for a given target cost it can achieve a lower occupancy, and viceversa.
For example, if we fix the average cost to $19.69$, the outage approach requires an average occupancy of $4.18\times10^5$, while the broadcast approach can achieve that same average cost with an average occupancy of $7.16\times10^4$, leading to a factor of $5.8\times$ reduction in average delay.
On the other hand, when fixing the average occupancy to be \eg $4.18\times10^5$, the outage approach requires and average cost of $19.69$, while the broadcast approach can reduce the cost to $17.73$ for the same average occupancy.

\subsection{Processing Flow Distribution}

In this section, we simulate the average processing input rate distribution for the $4$ functions and $56$ clients across the computing network nodes under DWCNC with both the outage approach and the broadcast approach, respectively shown in Fig. \ref{fig: result 4} and Fig. \ref{fig: result 5}. The average input rate for each client and each service is again equal to $0.7$ and we set the control parameter $V$ to $10^4$.

Observe from Figs. \ref{result4_1} and \ref{result5_1} that the implementation of function $(1,1)$ mostly concentrates at the APs (nodes $1,6,7$), motivated by the fact that the APs have cheaper processing resources than the UEs. Note, however, that part of the processing of function $(1,1)$ still takes place at the UEs, even though the APs still have available processing capacity. This results from the fact that, for certain clients $s \rightarrow d$, there exist short paths connecting node $s$ and $d$ not passing through any AP, such that commodity $(d,1,0)$ steadily gets routed along these paths and gets processed at the corresponding UEs, instead of getting routed along longer paths that pass through APs. 
Comparing Fig. \ref{result4_1} and Fig. \ref{result5_1}, it can be seen that the implementation of function $(1,1)$ concentrates even more at the APs when using the broadcast approach. This is  due to the enhanced transmission ability of the broadcast approach, which lowers the cost of taking longer paths passing through APs. 

\begin{figure*}[ht]
\centering
\begin{minipage}[t]{0.08\textwidth}
\centering \includegraphics[width=1.42in]{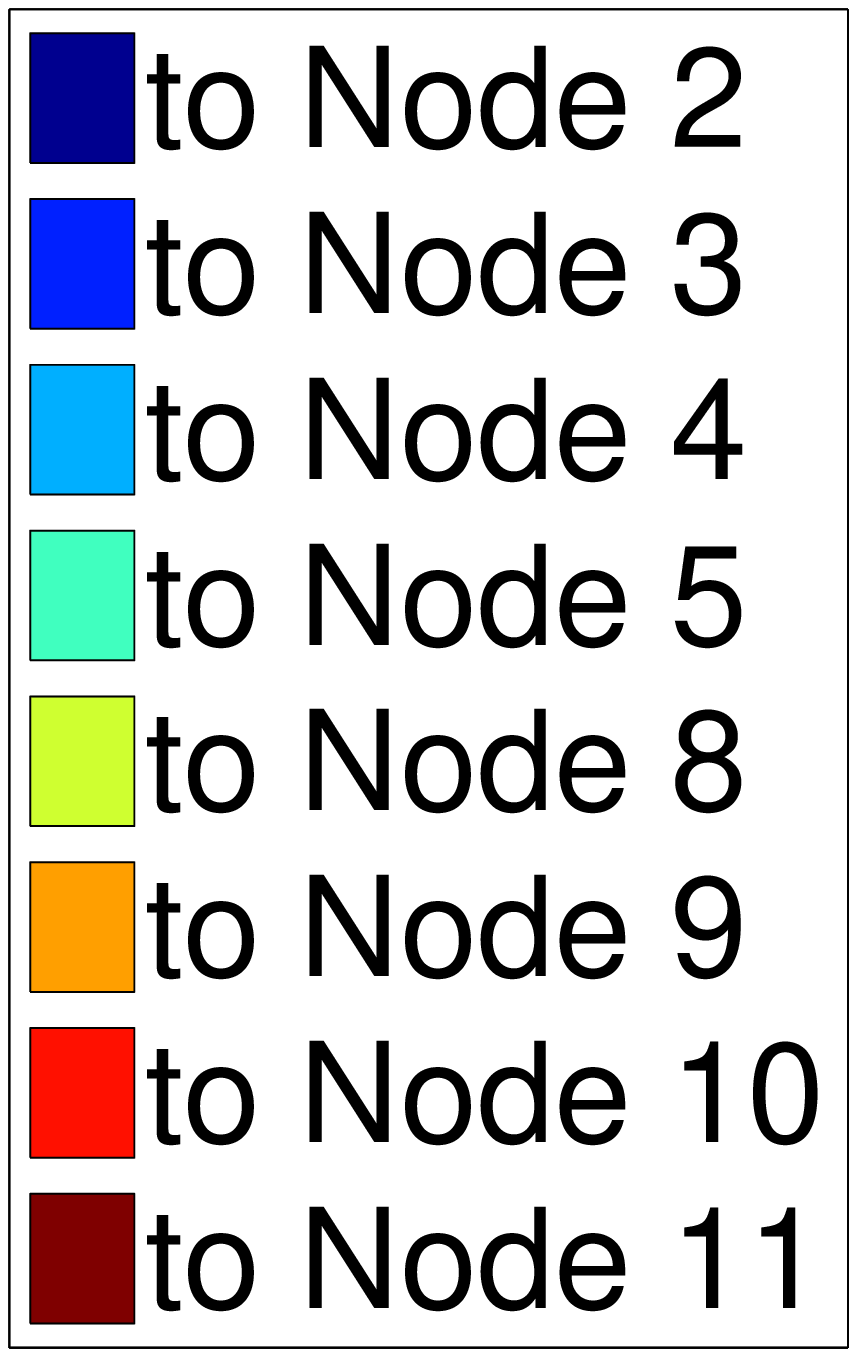}
\end{minipage}
\subfigure[]{
\centering \includegraphics[width=1.53in]{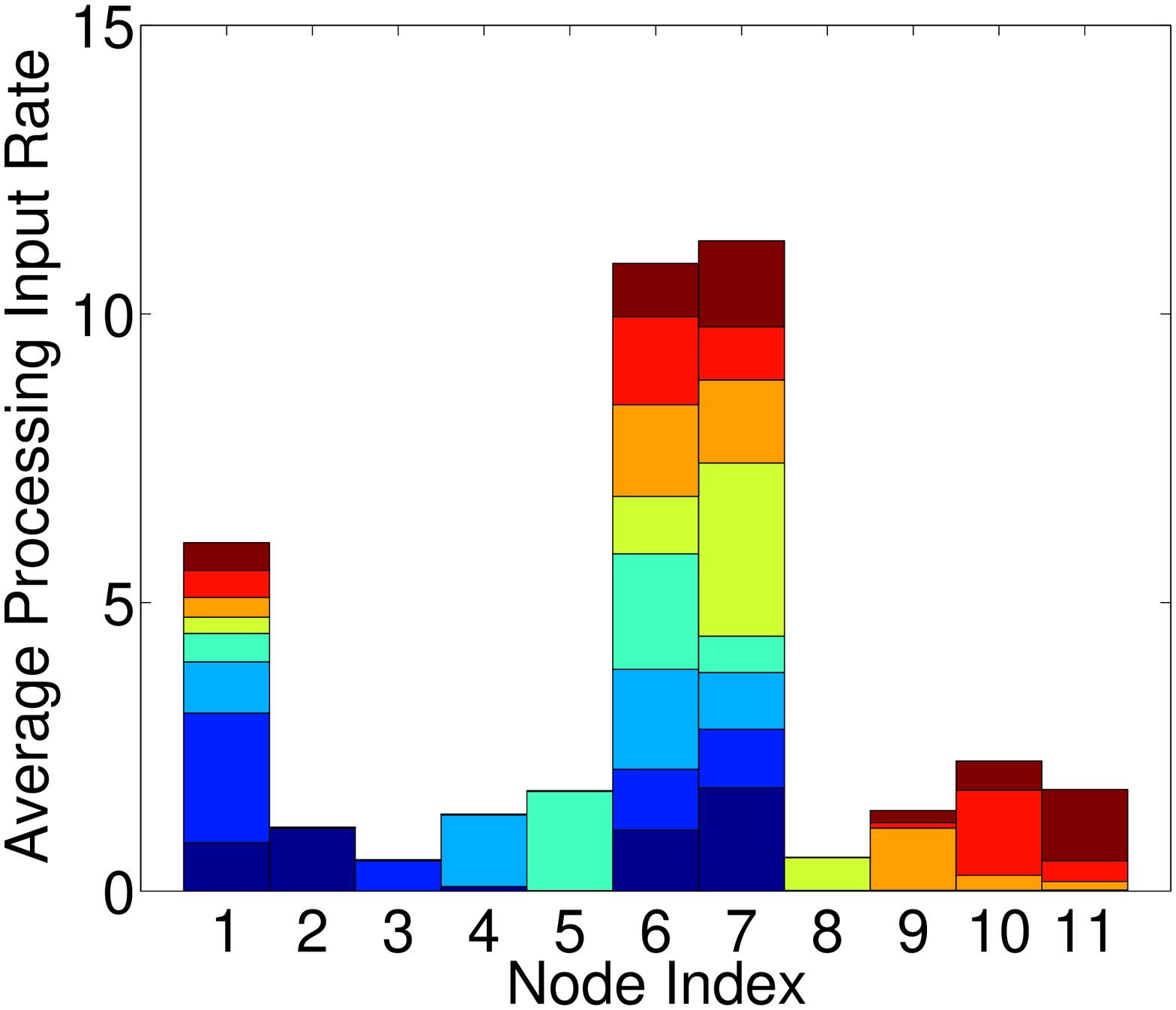}
\label{result4_1}
}
\hspace{-0.9cm}
\subfigure[]{
\centering \includegraphics[width=1.53in]{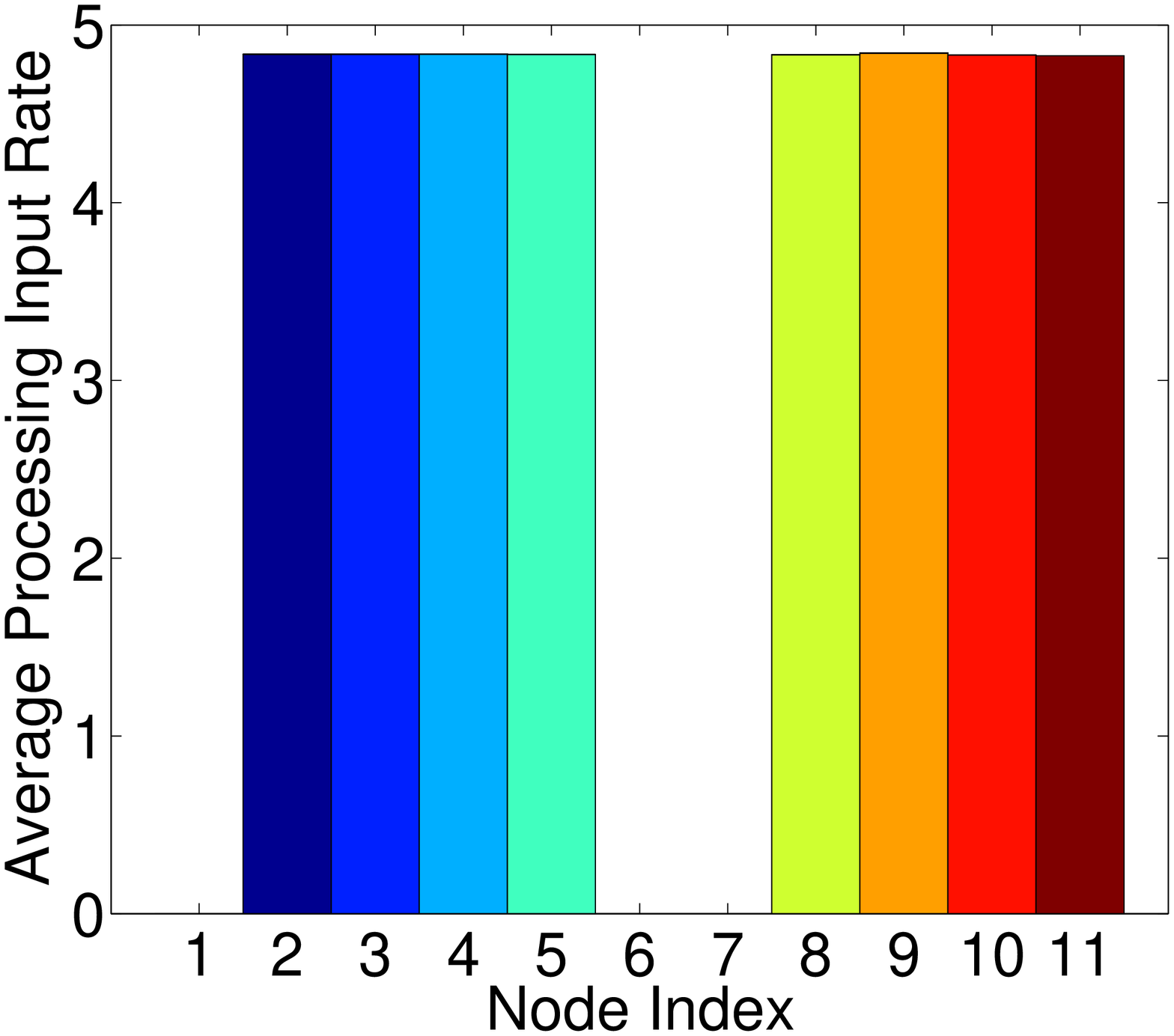}
\label{result4_2}
}
\hspace{-0.8cm}
\subfigure[]{
\centering \includegraphics[width=1.53in]{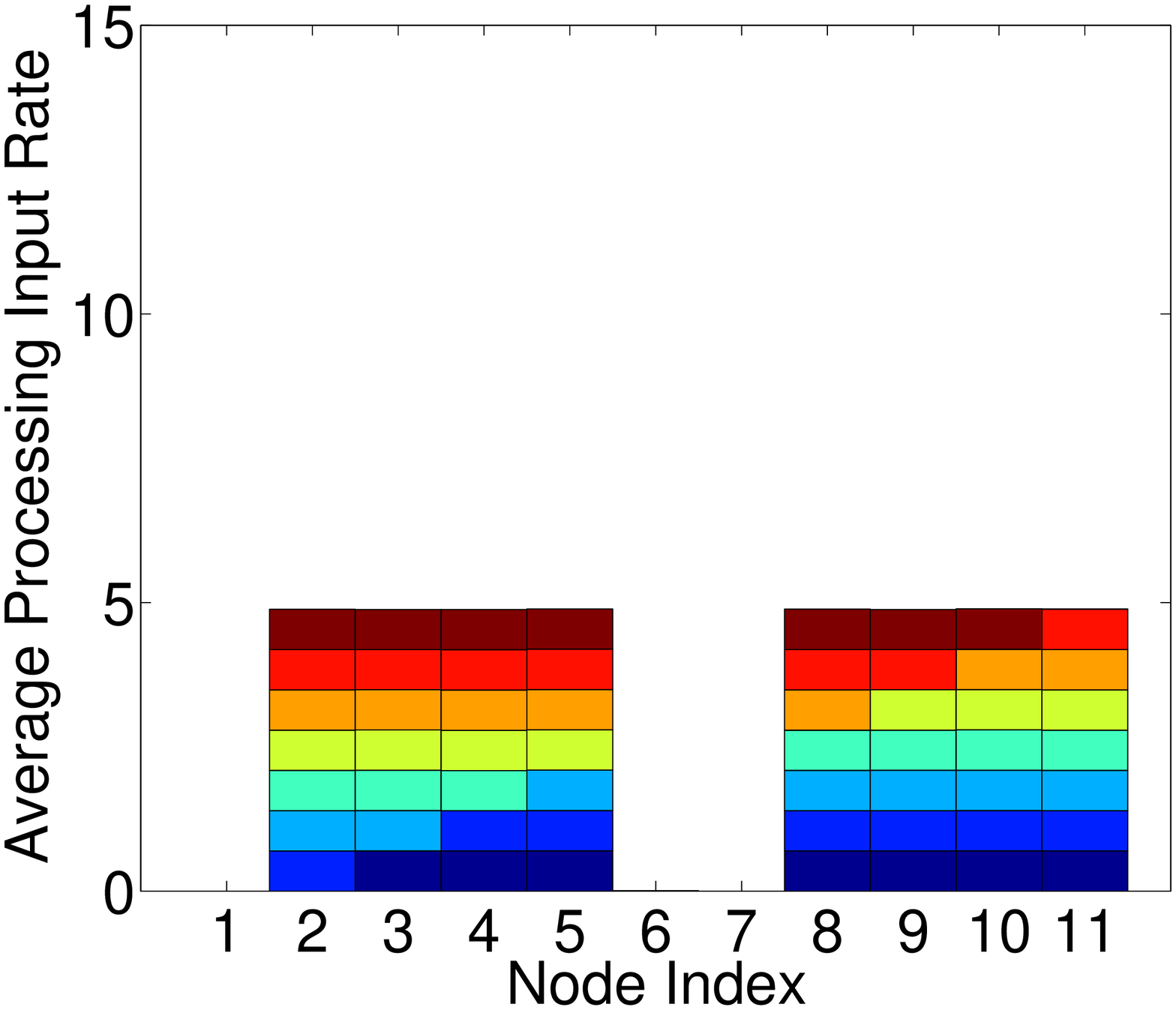}
\label{result4_3}
}
\hspace{-0.9cm}
\subfigure[]{
\centering \includegraphics[width=1.53in]{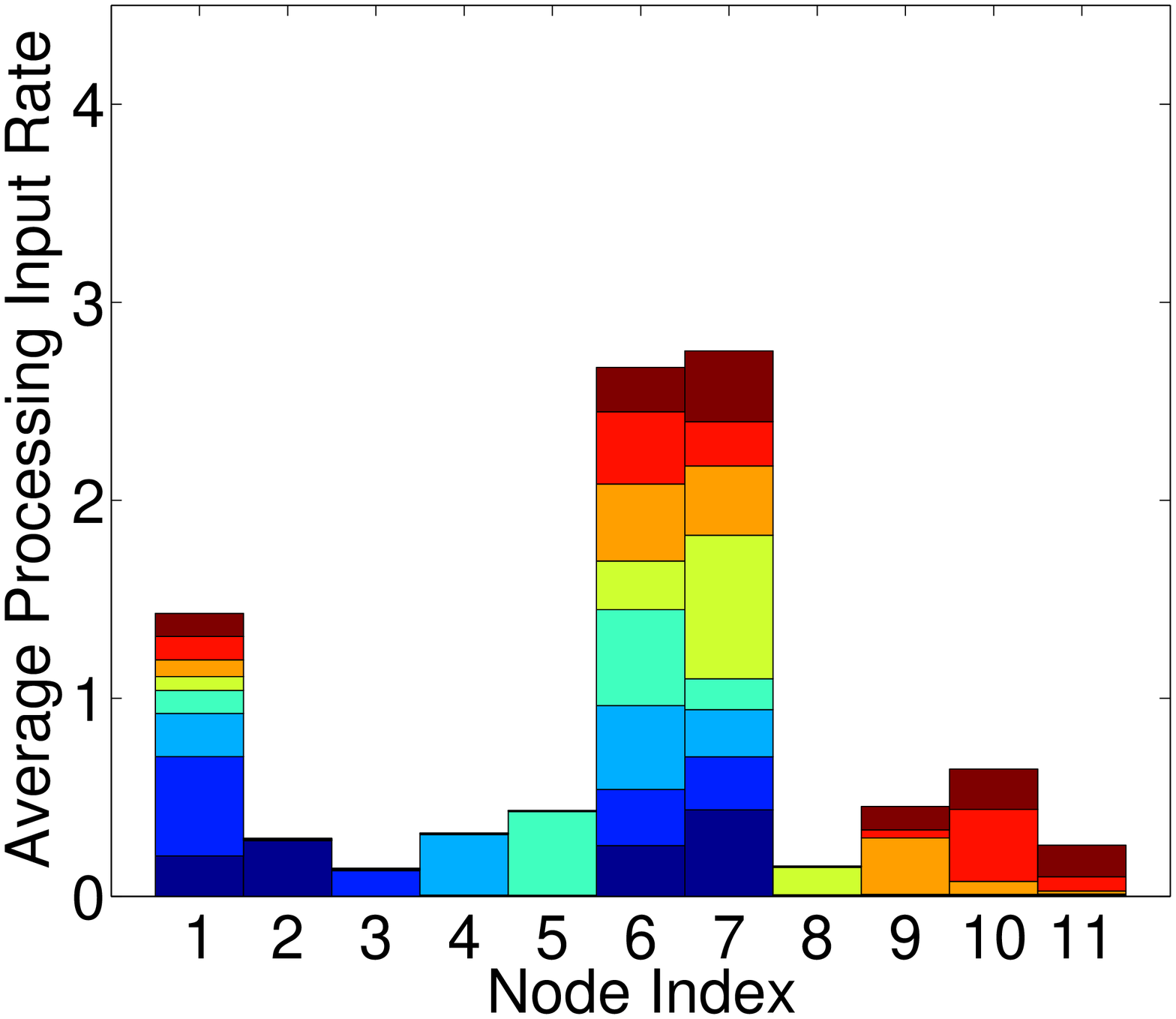}
\label{result4_4}
}

\caption{Average processing input rate distribution of DWCNC with the outage approach. a)~Service $1$, Function $1$; b)~Service $1$, Function $2$; c)~Service $2$, Function $1$; d)~Service $2$, Function $2$.}
\vspace{-0.5cm}
\label{fig: result 4}
\end{figure*}

\begin{figure*}[ht]
\centering
\begin{minipage}[t]{0.08\textwidth}
\centering \includegraphics[width=1.42in]{large_legend_distribution.eps}
\end{minipage}
\subfigure[]{
\centering \includegraphics[width=1.53in]{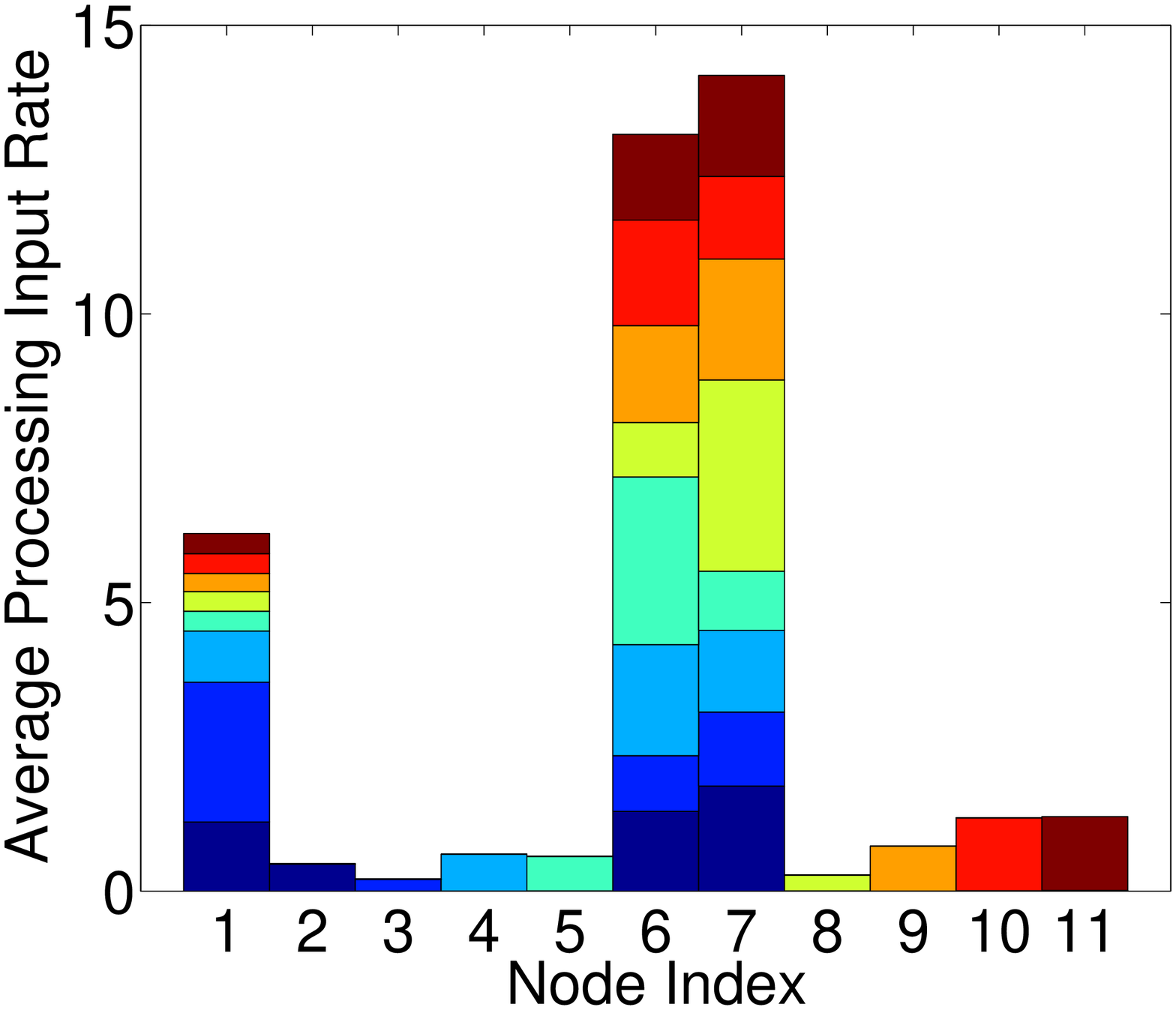}
\label{result5_1}
}
\hspace{-0.9cm}
\subfigure[]{
\centering \includegraphics[width=1.53in]{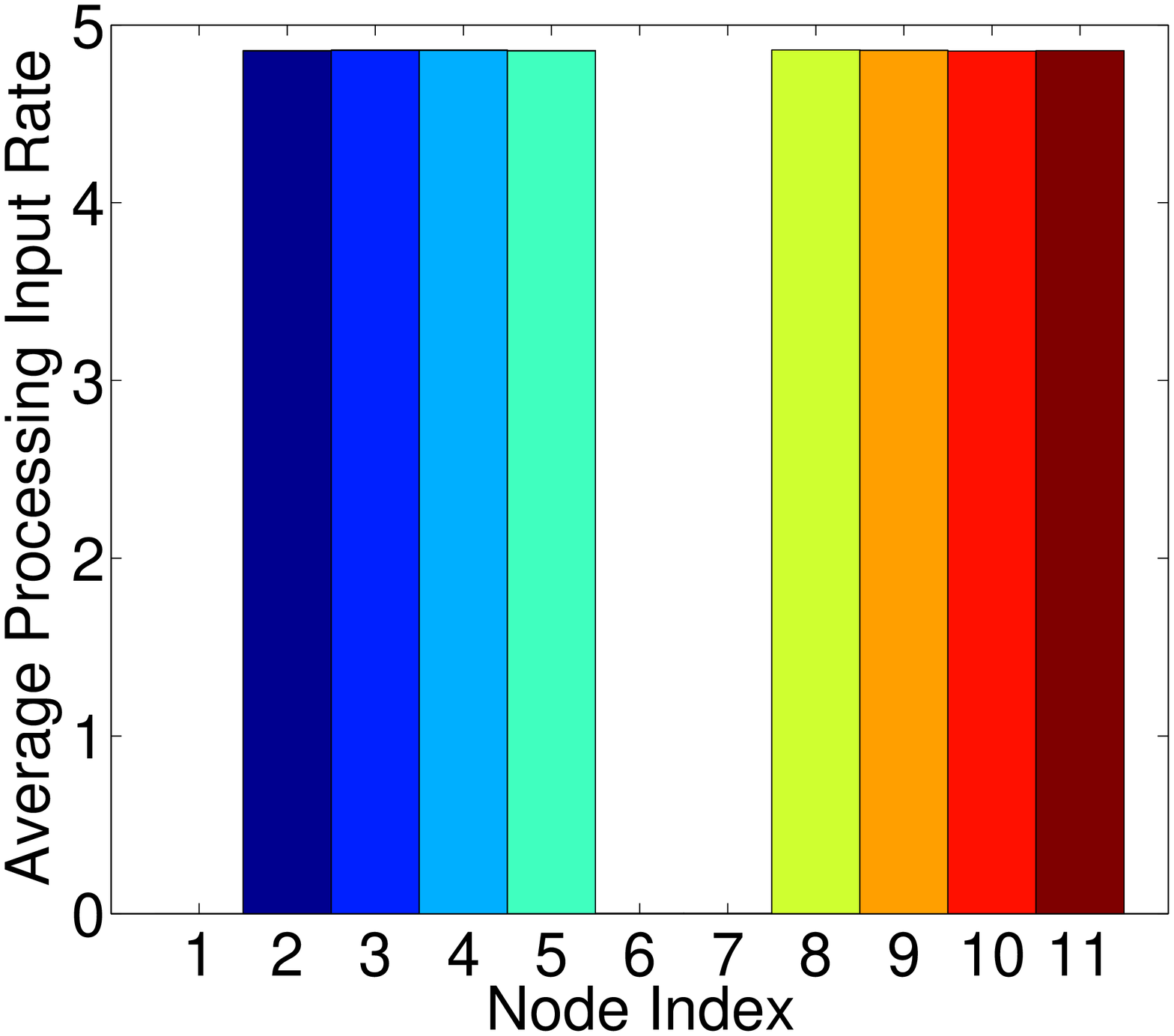}
\label{result5_2}
}
\hspace{-0.8cm}
\subfigure[]{
\centering \includegraphics[width=1.53in]{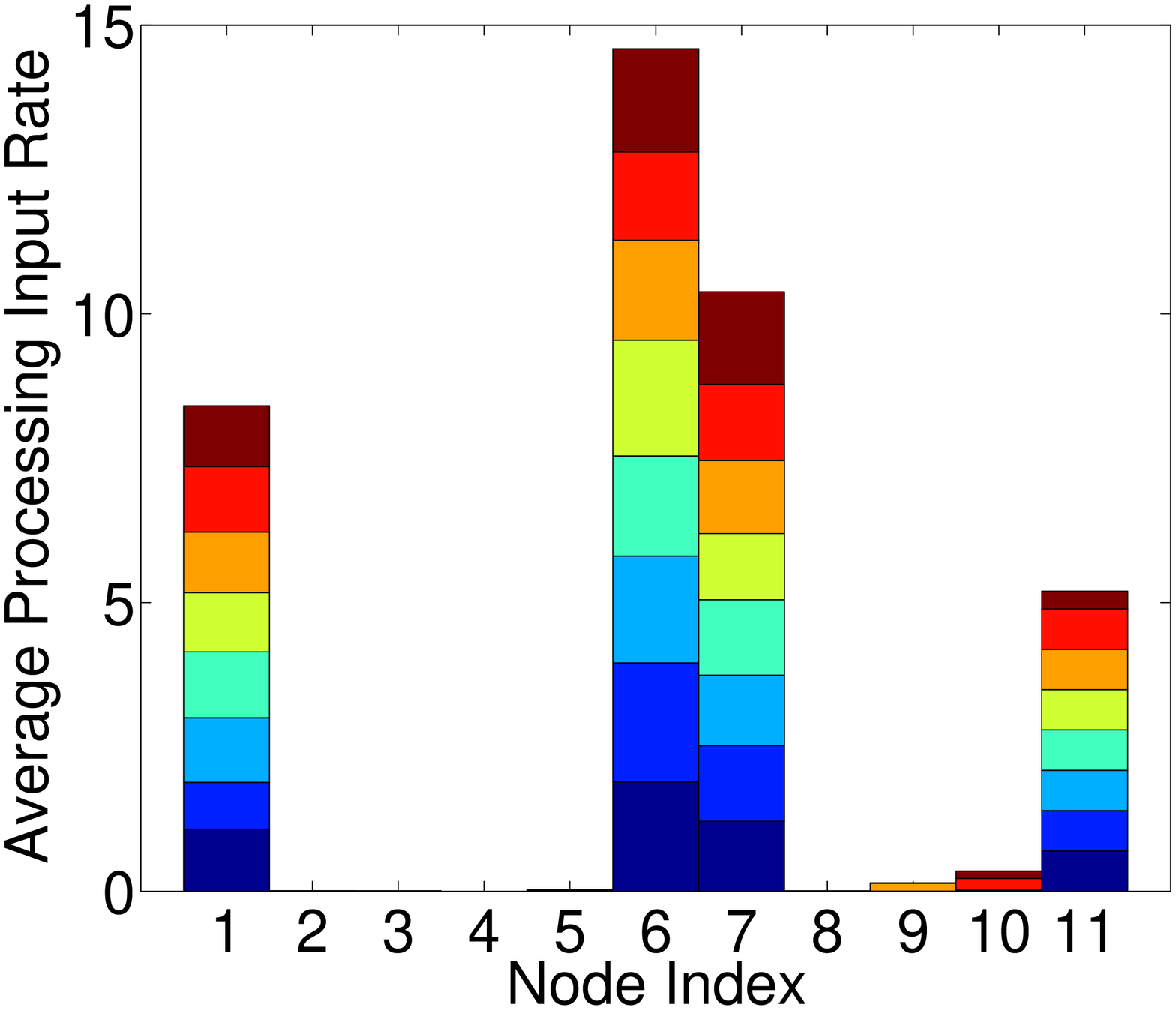}
\label{result5_3}
}
\hspace{-0.9cm}
\subfigure[]{
\centering \includegraphics[width=1.53in]{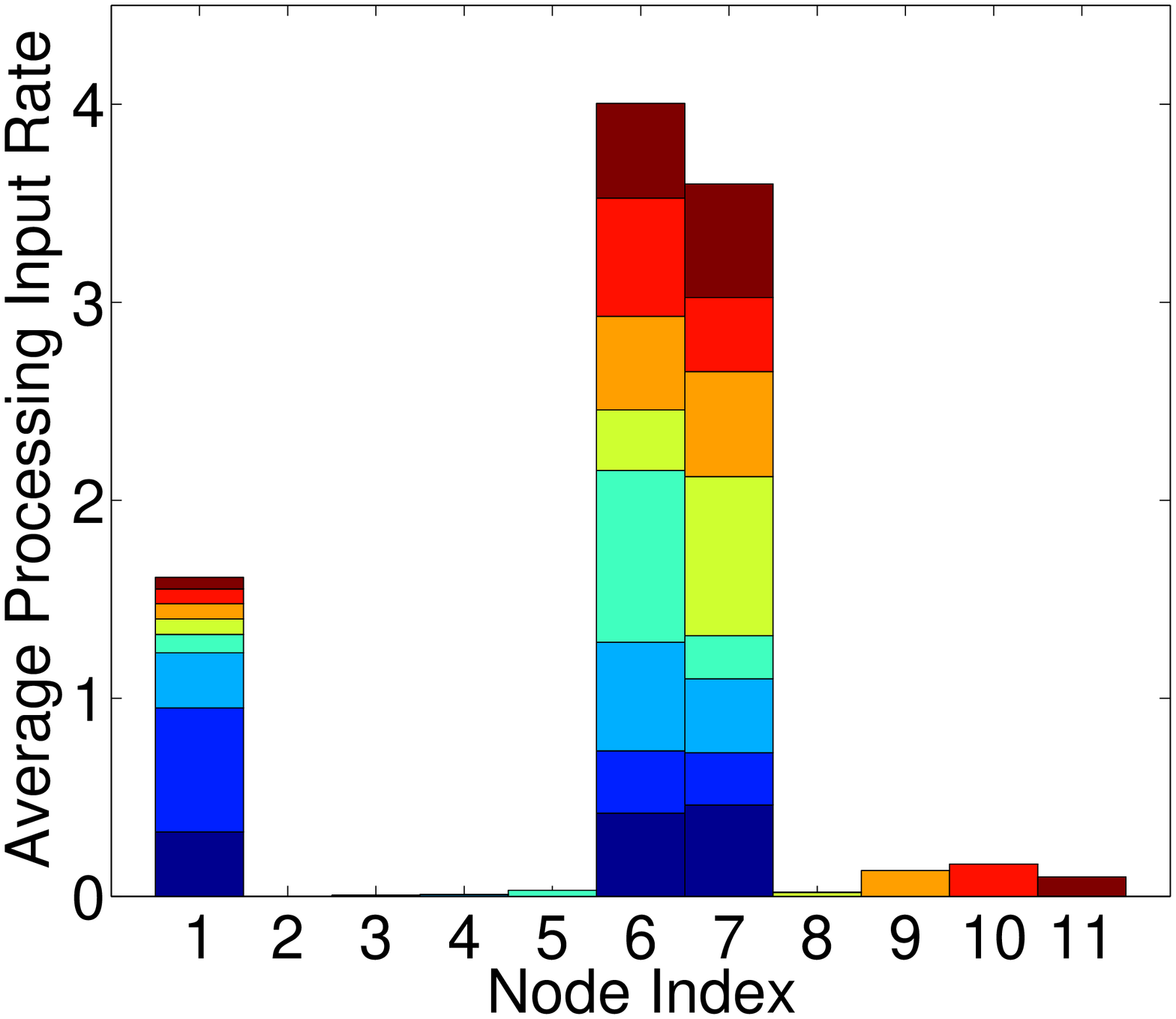}
\label{result5_4}
}

\caption{Average processing input rate distribution of DWCNC with the broadcast approach. a)~Service $1$, Function $1$; b)~Service $1$, Function $2$; c)~Service $2$, Function $1$; d)~Service $2$, Function $2$. }
\vspace{-0.5cm}
\label{fig: result 5}
\end{figure*}

Figs. \ref{result4_2} and \ref{result5_2} show the average processing input rate distribution of function $(1,2)$, which is an expansion function. As expected, the processing of  input commodity $(d,1,1)$ concentrates at its destination node $d$ when using both the broadcast and the outage approach. This results from DWCNC trying to minimize the transmission cost impact of the expanded-size commodities generated by function $(1,2)$.

For Service $2$, observe that the average processing input rate distribution of function $(2,1)$ is quite different depending on the coding scheme used, as illustrated in Figs. \ref{result4_3} and \ref{result5_3}. With the outage approach, Fig. \ref{result4_3} shows that function $(2,1)$, a compression function, is implemented at all the UEs except the destination node $d$, and at the APs. This is because, for each client $s\rightarrow d$, implementing function $(2,1)$ at the source node $s$ reduces the transmission cost of service $2$ by compressing the source commodity $(d,2,0)$ before entering the network. In contrast, as shown in Fig. \ref{result5_3}, the implementation of function $(2,1)$ using the broadcast approach mostly concentrates at the APs. This is again due to the increased transmission efficiency of the broadcast approach, which allows to push the processing of commodity $(d,2,0)$ to the cheaper APs with a smaller penalty in the transmission cost required to route the uncompressed commodity. 

The processing distribution of function $(2,2)$, shown in Figs. \ref{result4_4} and \ref{result5_4}, display a similar behavior as that of function $(1,1)$ shown in Figs. \ref{result4_1} and \ref{result5_1}. 
Note how, again, 
the processing distribution concentrates more on the APs when adopting the broadcast, illustrating, once more, how its enhanced transmission efficiency allows a better utilization of the cheaper processing nodes.

\section{Conclusion}
\label{conclusion}
We considered the problem of optimal distribution of augmented information services over wireless computing networks.
We characterized the capacity region of a wireless computing network
and designed
a dynamic wireless computing network control (DWCNC) algorithm that drives local transmissions-plus-processing flow scheduling and resource allocation decisions, shown to achieve arbitrarily close to minimum average network cost, while subject to network delay increase with the general trade off order $[O(1/V ), O(V )]$.
Our solution captures the unique chaining and flow scaling aspects of AgI services, while exploiting the use of the broadcast approach coding scheme over the wireless channel.



%

\appendices
\appendices
\section{Proof of Theorem \ref{thm: network_capacity_region}:necessity}
\label{appendix:_proof_of_theorem_1}
\subsection{Proof of Necessity}
We prove that \eqref{eq_thm1_stability}-\eqref{etas} are necessary for the stability of the wireless computing network, and that the minimum average cost can be achieved according to \eqref{eq_thm1_minimum_cost} and \eqref{eq:underline}.



Recall that our policy space includes policies that use multi-copy routing, which allow multiple copies of the same information unit to travel through the network. We say that two information units are {\em equivalent} if the successful delivery of one of them to its destination does not require the delivery of the other to satisfy the service demand. Note that equivalent information units may be exact copies of each other, but may also be distinct units that have evolved via service processing from a common copy. 

Let us assume that when an information unit of final commodity $(d,M)$ gets delivered to destination $d$, all other {\em equivalent} information units are immediately discarded from the network -- an ideal assumption for traffic reduction of algorithms with multi-copy routing.
We define $\mathcal I^{(d,m)}(t)$ as the set of information units of commodity $(d,m)$ that, after going through the sequence of service functions $\{m+1,m+2,\dots,M\}$,
are delivered to destination $d$ within the first $t$ timeslots.
Suppose there exists an algorithm that stabilizes the wireless computing network, possibly allowing multiple copies of a given information unit to flow through the network. Under this algorithm, define
\begin{itemize}
\item $I_i^{(d,m)}(t)$: the number of information units within $\mathcal I^{(d,m)}(t)$ 
that exogenously enter node $i$; 
\item $I_{i,\text{pr}}^{\left( {d,m} \right)}(t)$ and $I_{\text{pr},i}^{\left( {d,m} \right)}(t)$: the number of information units within $\mathcal I^{(d,m)}(t)$ 
that enter/exit the processing unit of node $i$; 
\item $I_{ij}^{\left( {d,m} \right)}(t)$: the number of times the information units within $\mathcal I^{(d,m)}(t)$ flow over link $(i, j)$. 
\end{itemize}

Since the algorithm stabilizes the network, we have, with probability 1,
\begin{equation}
\lim_{t \rightarrow \infty} \frac {\sum_{\tau=0}^t a_i^{(d,m)}(t)} {t} =  \lim_{t \rightarrow \infty}  \frac {I_i^{(d,m)}(t)}
{t} =\lambda_i^{(d,m)},\qquad \forall i,(d,m).
\label{eq_rate_stable_a_X}
\end{equation}


Moreover, the total number of arrivals (both exogenously and endogenously) to node $i$ of information units within $\mathcal I^{(d,m)}(t)$ must be equal to the number of departures from node $i$ of information units within $\mathcal I^{(d,m)}(t)$. Therefore, we have, for $i\neq d$ or $m<M$,
\begin{align}
\label{eq_packets_flow_balance}
&\sum\nolimits_{j:j\neq i} {I_{ji}^{\left( {d,m} \right)}} (t) + I_{\text{pr},i}^{\left( {d,m} \right)}(t) + I_i^{(d,m)}(t)  =\sum\nolimits_{j:j\neq i} {I_{ij}^{\left( {d,m} \right)}} (t) + I_{i,\text{pr}}^{\left( {d,m} \right)}(t),
\end{align}
and, for $m<M$ and for all $i$ and $d$,
\begin{equation}
\label{packet_flow_chaining}
I_{\text{pr},i}^{\left( {d,m+1} \right)}(t)=\xi^{(d,m+1)}I_{i,\text{pr}}^{\left( {d,m} \right)}(t). 
\end{equation}

Furthermore, on the one hand, define the following variables for transmission:
\begin{itemize}
\item $T ({\bf s},t)$: the number of timeslots within the first $t$ timeslots in which the network state is $\bf s$
\item $\tilde \alpha_{i,k}^{\text{tr}}({\tilde{\bf s}},t)$: the number of timeslots within the first $t$ timeslots in which $k$ transmission resource units are allocated at node $i$, while the previous network state (CSI feedback in the previous timeslot) is $\tilde{\bf s}$
\item $\tilde \beta_{i,\text{tr}}^{(d,m)}({\tilde{\bf s}},k,t)$: the accumulated time (in possibly fractional timeslots) during the first $t$ timeslots used by node $i$ to transmit information units within $\mathcal I^{(d,m)}(t)$, while $k$ resource units are allocated for transmission, and the previous network state is $\tilde{\bf s}$
\item $\rho_{i,{\bf s}}^{(d,m)}({\tilde{\bf s}},k,t)$: the accumulated time during the first $t$ timeslots used by node $i$ to transmit information units within $\mathcal I^{(d,m)}(t)$ when the network state is $\bf s$, while the network state in the previous timeslot is $\tilde {\bf s}$, and $k$ resource units are allocated for transmission
\item $\gamma_{i,n, \text{tr}}^{(d,m)}({\tilde{\bf s}},k, {\bf s}, t)$: the number of times during the first $t$ timeslots that an information unit within $\mathcal I^{(d,m)}(t)$ is transmitted by node $i$ with $k$ transmission resource units, and fall into the $n$-th partition, while the network state is $\bf s$ and the previous network state is $\tilde{\bf s}$
\item $\tilde \eta_{ij}^{(d,m)}({\tilde{\bf s}},k,{\bf s},n,t)$: the number of times during the first $t$ timeslots that an information unit within $\mathcal I^{(d,m)}(t)$ transmitted by node $i$ with $k$ transmission resource units, is retained by node $j$ while belonging to the $n$-th partition, when the network state is $\bf s$ and the previous network state is $\tilde{\bf s}$
\end{itemize}



Based on the above definitions and the transmission constraints, we have the following relations:
\begin{align}
\label{eq_bound_alpha_tr}
&\frac{{\tilde \alpha _{i,k}^{\text{tr}}({{\tilde {\bf s}}},t)}}{{{T_{\tilde{\bf s}}}\left( t \right)}} \ge 0,\ \ \sum\nolimits_{k = 0}^{K_i^{tr}} {\frac{{\tilde \alpha _{i,k}^{\text{tr}}({\tilde{\bf s}},t)}}{{{T_{\tilde {\bf s}}}\left( t \right)}}}  = 1,\qquad \forall i,\tilde {\bf s},t, \\
\label{eq_bound_beta_tr}
&\frac{{\tilde \beta _{i,\text{tr}}^{\left( {d,m} \right)}\left( {{\tilde{\bf s}},k,t} \right)}}{{\tilde \alpha _{i,k}^{\text{tr}}\left( {{\tilde{\bf s}},t} \right)}} \ge 0,\ \sum\nolimits_{\left( {d,m} \right)} {\frac{{\tilde \beta _{i,\text{tr}}^{\left( {d,m} \right)}\left( {{\tilde{\bf s}},k,t} \right)}}{{\tilde \alpha _{i,k}^{\text{tr}}\left( {{\tilde{\bf s}},t} \right)}} \le 1},\qquad \forall i,k,\tilde {\bf s},t, \\
\label{eq_bound_eta}
&\frac{{\tilde \eta _{ij}^{(d ,m)}({\tilde{\bf s}},k,{\bf s}, n,t)}}{{\gamma _{i,n,\text{tr}}^{(d,m)}({\tilde{\bf s}},k,{\bf s},t)}} \ge 0,\ \sum\nolimits_{j \in {\Omega _{i,n}}} {\frac{{\tilde \eta _{ij}^{(d ,m)}({\tilde{\bf s}},k,{\bf s}, n,t)}}{{\gamma _{i,n,\text{tr}}^{(d,m)}({\tilde{\bf s}},k,{\bf s}, t)}} \le 1},\qquad \forall i,d,m,\tilde {\bf s},{\bf s},t,
\end{align}
where we define $0/0=1$ for any term on the denominator happen to be zero. For each link $(i,j)$, each commodity $(d,m)$, and all $t$, we then have
\begin{align}
\label{eq_time_average_rate_tr}
\frac{{I_{ij}^{(d,m)}\left( t \right)}}{t} = &\sum\nolimits_{{\bf{\tilde s}} \in \mathcal S} {\frac{{{T_{{\bf{\tilde s}}}}\left( t \right)}}{t}\sum\nolimits_{k = 0}^{K_i^{\text{tr}}} {\frac{{\tilde \alpha _{i,k}^{\text{tr}}({\bf{\tilde s}},t)}}{{{T_{{\bf{\tilde s}}}}\left( t \right)}}\frac{{\tilde \beta _{i,\text{tr}}^{\left( {d,m} \right)}\left( {{\bf{\tilde s}},k,t} \right)}}{{\tilde \alpha _{i,k}^{\text{tr}}({\bf{\tilde s}},t)}}\sum\limits_{{\bf{s}} \in \mathcal S} {\frac{{\rho _{i,{\bf{s}}}^{\left( {d,m} \right)}\left( {{\bf{\tilde s}},k,t} \right)}}{{\tilde \beta _{i,\text{tr}}^{\left( {d,m} \right)}\left( {{\bf{\tilde s}},k,t} \right)}}} } } \notag\\
&\times \sum\nolimits_{n = 1}^{g_{i,{\bf{s}}}^{ - 1}\left( j \right)} {\frac{{\gamma _{i,n,\text{tr}}^{(d,m)}({\bf{\tilde s}},k,{\bf{s}},t)}}{{\rho _{i,{\bf{s}}}^{\left( {d,m} \right)}\left( {{\bf{\tilde s}},k,t} \right)}}} \frac{{\tilde \eta _{ij}^{(d ,m)}({\tilde {\bf{s}}},k,{\bf{s}},n,t)}}{{\gamma _{i,n,\text{tr}}^{(d,m)}({\bf{\tilde s}},k,{\bf{s}},t)}}.
\end{align}
The network state process yields,
\begin{equation}
\label{eq_stationary_prob}
\mathop {\lim }\limits_{t \to \infty } \frac{{{T_{\bf{s}}}\left( t \right)}}{t} = {\pi _{\bf{s}}},\ \mathrm{with\ prob.}\ 1,
\end{equation}
and due to fact that $y_{i,k}^{\rm{tr}}(\tau)$ is independent of ${\bf S}(\tau)$ given ${\bf S}(\tau-1)=\tilde {\bf s}$, we also have
\begin{equation}
\mathop {\lim }\limits_{t \to \infty } \frac{{\rho _{i,{\bf{s}}}^{\left( {d,m} \right)}\left( {{\bf{\tilde s}},k,t} \right)}}{{\tilde \beta _{i,\text{tr}}^{\left( {d,m} \right)}\left( {{\bf{\tilde s}},k,t} \right)}} = {P_{{\bf{\tilde ss}}}},\qquad \forall i,d,m,
\end{equation}
where $P_{{\bf{\tilde ss}}}\triangleq \Pr({\bf S}(t)={\bf s}|{\bf S}(t-1)=\tilde {\bf s})$.
In addition, we upper bound the average rate of the $n$-th partition as follows:
\begin{align}
\label{eq_cap_achieving_rate_tr}
0\!\le\frac{{\gamma _{i,n,\text{tr}}^{(d,m)}\!({\bf{\tilde s}},k,{\bf{s}},t)}}{{\rho _{i,{\bf{s}}}^{\left( {d,m} \right)}\left( {{\bf{\tilde s}},k,t} \right)}}\! \le\! {R_{i,{g_{i,n}},k}}\!\left( {\bf{s}} \right)\! -\! {R_{i,{g_{i,n - 1}},k}}\!\left( {\bf{s}} \right)\!, \qquad \forall i,k,d,m,\tilde {\bf s}, {\bf s},t.
\end{align}

On the other hand, define the following variables for processing:
\begin{itemize}
\item $\alpha_{i,k}^{\text{pr}}(t)$: the number of timeslots during the first $t$ timeslots in which node $i$ allocates $k$ processing resource units
\item $\beta_{i,\text{pr}}^{(d,m)}(k,t)$: the accumulated time used by node $i$ to process information units within $\mathcal I^{(d,m)}(t)$, while $k$ resource units are allocated for processing
\item $\gamma_{i,\text{pr}}^{(d,m)}(k,t)$: the number of information units within $\mathcal I^{(d,m)}(t)$ that are processed by node $i$ with $k$ processing resource units during the first $t$ timeslots
\end{itemize}
Based on the above definitions and the processing constraints, we have the following relations:
\begin{align}
\label{eq_bound_alpha_pr}
&\frac{{\alpha _{i,k}^{\text{pr}}(t)}}{t} \ge 0,\;\;\sum\nolimits_{k = 0}^{K_i^{pr}} {\frac{{\alpha _{i,k}^{\text{pr}}(t)}}{t}}  = 1,\qquad \forall i,t,\\
\label{eq_bound_beta_pr}
&\frac{{\beta _{i,\text{pr}}^{\left( {d,m} \right)}\left( {k,t} \right)}}{{\alpha _{i,k}^{\text{pr}}\left( t \right)}} \ge 0,\ \sum\nolimits_{\left( {d,m} \right)} {\frac{{\beta _{i,\text{pr}}^{\left( {d,m} \right)}\left( {k,t} \right)}}{{\alpha _{i,k}^{\text{pr}}\left( t \right)}} \le 1},\qquad \forall i,k,t,\\
\label{eq_bound_processing_capacity}
&0\le \frac{{\gamma _{i,\text{pr}}^{(d,m)}(k,t)}}{{\beta _{i,\text{pr}}^{\left( {d,m} \right)}\left( {k,t} \right)}} \le \frac{{R_{i,k}}}{r^{(d,m+1)}},\qquad \forall i,t,k,d,m<M.
\end{align}
For each node $i$, we then have, for all $i$, $(d,m)$ and $t$,
\begin{equation}
\label{eq_time_average_rate_pr}
\frac{{I_{i,\text{pr}}^{(d,m)}\left( t \right)}}{t} = \sum\nolimits_{k = 0}^{K_i^{\text{pr}}} {\frac{{\alpha _{i,k}^{\text{pr}}(t)}}{t}} \frac{{\beta _{i,\text{pr}}^{\left( {d,m} \right)}\left( {k,t} \right)}}{{\alpha _{i,k}^{\text{pr}}(t)}}\frac{{\gamma _{i,\text{pr}}^{(d,m)}(k,t)}}{{\beta _{i,\text{pr}}^{\left( {d,m} \right)}\left( {k,t} \right)}}.
\end{equation}

Because the constraints in \eqref{eq_bound_alpha_tr}-\eqref{eq_bound_eta}, \eqref{eq_cap_achieving_rate_tr}, and \eqref{eq_bound_alpha_pr}-\eqref{eq_bound_processing_capacity} define bounded ratio sequences with finite dimensions, there exists an infinitely long subsequence of timeslots $\{t_u\}$ over which the time average cost achieves its $\liminf$ value $\underline h$, and the ratio terms converge:
\begin{align}
&\Scale[1]{\mathop {\lim }\limits_{{t_u} \to \infty } \frac{1}{{{t_u}}}\sum\nolimits_{\tau  = 0}^{{t_u} - 1} {h\left( \tau  \right)}  = \underline h,\ \mathop {\lim }\limits_{{t_u} \to \infty } \frac{{\tilde \alpha _{i,k}^{\text{tr}}({\bf{\tilde s}},{t_u})}}{{{T_{{\bf{\tilde s}}}}\left( {{t_u}} \right)}} = \alpha _{i,k}^{\text{tr}}({\bf{\tilde s}}),\ \mathop {\lim }\limits_{{t_u} \to \infty } \frac{{\tilde \beta _{i,\text{tr}}^{\left( {d,m} \right)}\left( {{\bf{\tilde s}},k,t_u} \right)}}{{\tilde \alpha _{i,k}^{\text{tr}}({\bf{\tilde s}},t_u)}} = \tilde \beta _{i,\text{tr}}^{\left( {d,m} \right)}\left( {{\bf{\tilde s}},k} \right),}\notag\\
&\Scale[1]{\mathop {\lim }\limits_{{t_u} \to \infty }\! \frac{{\tilde \eta _{ij}^{(d,\phi ,m)}({\tilde{\bf{s}}},k,{\bf{s}},n,t_u)}}{{\gamma _{i,n,\text{tr}}^{(d,m)}({\bf{\tilde s}},k,{\bf{s}},t_u)}} \!=\! \tilde \eta _{ij}^{(d,\phi ,m)}({\tilde{\bf{s}}},k,{\bf{s}},n),\ \mathop {\lim }\limits_{t_u\rightarrow \infty }\!\frac{{\gamma _{i,n,\text{tr}}^{(d,m)}\!({\bf{\tilde s}},k,{\bf{s}},t_u)}}{{\rho _{i,{\bf{s}}}^{\left( {d,m} \right)}\left( {{\bf{\tilde s}},k,t_u} \right)}}\!=\!F^{(d,m)}_{ij}(k,{\bf s}),}\notag\\
&\Scale[1]{\mathop {\lim }\limits_{{t_u} \to \infty }\! \frac{{\alpha _{i,k}^{\text{pr}}({t_u})}}{t} \!=\! \alpha _{i,k}^{\text{pr}},\ \mathop {\lim }\limits_{{t_u} \to \infty }\! \frac{{\beta _{i,\text{pr}}^{\left( {d,m} \right)}\left( {k,{t_u}} \right)}}{{\alpha _{i,k}^{\text{pr}}\left( {{t_u}} \right)}} \!=\! \beta _{i,\text{pr}}^{\left( {d,m} \right)}\left( k \right),\ \mathop {\lim }\limits_{{t_u} \to \infty }\! \frac{{\gamma _{i,\text{pr}}^{(d,m)}(k,t_u)}}{{\beta _{i,\text{pr}}^{\left( {d,m} \right)}\left( {k,t_0} \right)}} \!=\! F^{(d,m)}_{i,\text{pr}}(k).}\notag
\end{align}
Define $f_{ij}^{(d,m)}(t)\triangleq \left.I_{ij}^{(d,m)}\left( t \right)\right/t$.


Then, it follows from \eqref{eq_time_average_rate_tr} that
\begin{align}
\label{eq_transmission_flow_rate}
&f_{ij}^{\left( {d,m} \right)} \triangleq \mathop {\lim }\limits_{{t_u} \to \infty } f_{ij}^{\left( {d,m} \right)}\left( {{t_u}} \right)\notag\\
&\buildrel (a) \over \le\sum\limits_{{\bf{\tilde s}} \in \mathcal S} {{\pi _{{\bf{\tilde s}}}}\sum\limits_{k = 0}^{K_i^{\text{tr}}} {\tilde \alpha _{i,k}^{\text{tr}}({\bf{\tilde s}})\tilde \beta _{i,\text{tr}}^{\left( {d,m} \right)}\left( {{\bf{\tilde s}},k} \right)\sum\limits_{{\bf{s}} \in S} {{P_{{\bf{\tilde ss}}}} } } } \sum\limits_{n = 1}^{g_{i,{\bf{s}}}^{ - 1}\left( j \right)} {\left[ {{R_{i,{g_{i,n}},k}}\left( {\bf{s}} \right) - {R_{i,{g_{i,n - 1}},k}}\left( {\bf{s}} \right)} \right]\tilde \eta _{ij}^{(d ,m)}({\bf{\tilde s}},k,{\bf{s}},n)}\notag\\
&\buildrel (b) \over =\sum\nolimits_{{\bf{s}} \in S} {{\pi _{\bf{s}}}\sum\nolimits_{k = 0}^{K_i^{\text{tr}}} {\alpha _{i,k}^{\text{tr}}({\bf{s}})\beta _{i,\text{tr}}^{\left( {d,m} \right)}\left( {{\bf{s}},k} \right)}} \sum\nolimits_{n = 1}^{g_{i,{\bf{s}}}^{ - 1}\left( j \right)} {\left[ {{R_{i,{g_{i,n}}}}\left( {\bf{s}} \right) - {R_{i,{g_{i,n - 1}}}}\left( {\bf{s}} \right)} \right]\eta _{ij}^{(d,m)}({\bf{s}},k,n)},
\end{align}
where inequality $(a)$ holds true due to the above converging terms for transmission, the convergence terms in \eqref{eq_stationary_prob}, and the fact that $F^{(d,m)}_{ij}(k,{\bf s})\le {R_{i,{g_{i,n}},k}}\!\left( {\bf{s}} \right)\! -\! {R_{i,{g_{i,n - 1}},k}}\!\left( {\bf{s}} \right)$; equality $(b)$ holds true due to the fact that $\pi_{\bf s} = \sum\nolimits_{\tilde {\bf s}\in \mathcal S}{\pi_{\tilde {\bf s}}P_{{\tilde{\bf s}}{\bf s}}}$ and the following definitions:
\begin{align}
&\alpha _{i,k}^{\text{tr}}({\bf{s}}) \triangleq \sum\nolimits_{{\bf{\tilde s}} \in S} {\frac{{{\pi _{{\bf{\tilde s}}}}{P_{{\bf{\tilde ss}}}}}}{{{\pi _{\bf{s}}}}}\tilde \alpha _{i,k}^{\text{tr}}({\bf{\tilde s}})},\quad \beta _{i,\text{tr}}^{\left( {d,m} \right)}\left( {{\bf{s}},k} \right) \triangleq \sum\nolimits_{{\bf{\tilde s}} \in S} {\frac{{{\pi _{{\bf{\tilde s}}}}{P_{{\bf{\tilde ss}}}}\tilde \alpha _{i,k}^{\text{tr}}({\bf{\tilde s}})}}{{{\pi _{\bf{s}}}\alpha _{i,k}^{\text{tr}}({\bf{s}})}}\tilde \beta _{i,\text{tr}}^{\left( {d,m} \right)}\left( {{\bf{\tilde s}},k} \right)}, \notag\\
&\eta _{ij}^{(d,m)}\!({\bf{s}},k,n)\! \triangleq\! \sum\nolimits_{{\bf{\tilde s}} \in S}\! {\frac{{{\pi _{{\bf{\tilde s}}}}{P_{{\bf{\tilde ss}}}}\tilde \alpha _{i,k}^{\text{tr}}({\bf{\tilde s}})\tilde \beta _{i,\text{tr}}^{\left( {d,m} \right)}\!\left( {{\bf{\tilde s}},k} \right)}}{{{\pi _{\bf{s}}}\alpha _{i,k}^{\text{tr}}({\bf{s}})\beta _{i,\text{tr}}^{\left( {d,m} \right)}\!\left( {{\bf{s}},k} \right)}}} \tilde \eta _{ij}^{(d,m)}\!({\bf{\tilde s}},k,{\bf{s}},n).\notag
\end{align}
In addition, define $f_{i,\text{pr}}^{(d,m)}(t)\triangleq \left.I_{i,\text{pr}}^{(d,m)}\left( t \right)\right/t$. With the converging terms for processing and the fact that $F^{(d,m)}_{i,\text{pr}}(k)\le \left.R_{i,k}\right/r^{(m+1)}$, for $m<M$, it follows from \eqref{eq_time_average_rate_pr} that
\begin{align}
f_{i,\text{pr}}^{\left( {d,m} \right)} & \triangleq \mathop {\lim }\limits_{{t_u} \to \infty } f_{i,\text{pr}}^{\left( {d,m} \right)}\left( {{t_u}} \right)\le {\sum\nolimits_{k = 0}^{K_i^{\text{pr}}} {\alpha _{i,k}^{\text{pr}}\beta _{i,tr}^{\left( {d,m} \right)}\left( k \right)\frac{{R_{i,k}}}{{{r^{\left( {m + 1} \right)}}}} } } .
\end{align}
Moreover, the flow efficiency and non-negativity constraints follow:
\begin{align}
&f_{i,\text{pr}}^{\left( {d,{M }} \right)}= 0,\ f_{\text{pr},i}^{\left( {d,0} \right)} = 0,\ f_{dj}^{\left( {d,{M }} \right)} = 0,\ f_{i,\text{pr}}^{\left( {d,m} \right)} \geq 0,\ f_{ij}^{\left( {d,m} \right)} \geq 0.
\end{align}
Furthermore, dividing by $t_u$ on both sides of \eqref{eq_packets_flow_balance} and \eqref{packet_flow_chaining}, and letting ${t_u}\rightarrow \infty$, we have, for $i\neq d$ or $m<M$, with the result of \eqref{eq_rate_stable_a_X},
\begin{equation}
\label{eq_flow_conservation}
\sum\nolimits_{j } \!{f_{ji}^{\left( {d,m} \right)}} \!\! + f_{\text{pr},i}^{\left( {d,m} \right)} \!\! + \lambda_i^{(d,m)}  =  \sum\nolimits_{j} \! {f_{ij}^{\left( {d,m} \right)}}  \!\! + f_{i,\text{pr}}^{\left( {d,m} \right)},
\end{equation}
and, for $m<M$ and all $i$ and $d$, 
\begin{equation}
\label{eq_flow_chaining}
f_{\text{pr},i}^{(m+1)} = \xi^{(m+1)}f_{i,\text{pr}}^{(m)}.
\end{equation}


Finally, the time average cost satisfies
\begin{align}
&\underline h 
= \mathop {\lim }\limits_{{t_u} \to \infty } \sum\nolimits_{i \in N} {\left[ {\sum\nolimits_{k = 0}^{K_i^{\text{pr}}} {\frac{{\alpha _{i,k}^{\text{pr}}({t_u})}}{{{t_u}}}} w_{i,k}^{\text{pr}} + } \right.} \left. {\sum\nolimits_{{\bf{\tilde s}} \in S} {\frac{{{T_{{\bf{\tilde s}}}}\left( {{t_u}} \right)}}{{{t_u}}}\sum\nolimits_{k = 0}^{K_i^{\text{tr}}} {\frac{{\tilde \alpha _{i,k}^{\text{tr}}({\tilde {\bf{ s}}},{t_u})}}{{{T_{{\tilde {\bf{ s}}}}}\left( {{t_u}} \right)}}w_{i,k}^{\text{tr}}} } } \right]\nonumber\\
&\buildrel (a) \over = \sum\nolimits_{i \in N}\! {\left( {\sum\nolimits_{k = 0}^{K_i^{\text{pr}}} {\alpha _{i,k}^{\text{pr}}} w_{i,k}^{\text{pr}} + \sum\nolimits_{k = 0}^{K_i^{\text{tr}}} {w_{i,k}^{\text{tr}}\sum\nolimits_{{\bf{s}} \in S} {{\pi _{\bf{s}}}\alpha _{i,k}^{\text{tr}}\left( {\bf{s}} \right)} }  } \right)},
\end{align}
where in $(a)$ we used the fact that $\sum\nolimits_{{\tilde {\bf s}}\in \mathcal S} \pi_{\tilde {\bf{ s}}}\tilde \alpha^{\text{tr}}_{i,k}(\tilde {\bf { s}})=\sum\nolimits_{ {\bf s}\in \mathcal S} \pi_{\bf{ s}}\alpha^{\text{tr}}_{i,k}({\bf { s}})$.

In summary, given $\{\lambda_i^{(d,m)}\}\in \Lambda$, this proves that there exists a set of flow variables and probability values that satisfy the constraints in Theorem \ref{thm: network_capacity_region}.
The minimum average cost $\overline h^*$ follows from taking the minimum of $\underline h$ over all the variable sets that stabilize the network.

\subsection{Proof of Sufficiency}
Given exogenous input rate matrix $\{\lambda_i^{(d,m)}  + \epsilon\}$, $\epsilon>0$, probability values $ \alpha_{i,k}^{\text{tr}}( {\bf s})$, $ \beta_{i,\text{tr}}^{(d,m)}( {\bf s},k)$, $\eta_{ij}^{(d,m)}({\bf s},k,n)$, $\alpha_{i,k}^{\text{pr}}$, $\beta_{i,\text{pr}}^{(d,m)}(k)$, and
multi-commodity flow variables $f_{ij}^{(d,m)}$,  $f_{i,\text{pr}}^{(d,m)}$, $f_{\text{pr},i}^{(d,m)}$ satisfying
(\ref{eq_thm1_stability})-(\ref{eq:underline}), we
construct a stationary randomized policy using single-copy routing such that:
\begin{align}
&\E \left\{ \mu_{ij}^{\left( {d ,m} \right)} (t)\right\} =  f_{ij}^{(d,m)},\ \E \left\{ \mu_{i,\text{pr}}^{\left( {d ,m} \right)} (t)\right\} =f_{i,\text{pr}}^{(d,m)},\ \E \left\{ \mu_{\text{pr},i}^{\left( {d ,m} \right)} (t)\right\} =f_{\text{pr},i}^{(d,m)}\label{eq2},
\end{align}
where
$\mu_{ij}^{\left( {d,m} \right)} (t)$, $\mu_{i,\text{pr}}^{(d,m)}(t)$, and $\mu_{\text{pr},i}^{(d,m)}(t)$ respectively
denote the flow rates assigned by the stationary randomized policy for transmission and processing. 
Plugging $\{\lambda_i^{(d,m)}  + \epsilon\}$ and the terms in \eqref{eq2} into (\ref{eq_thm1_stability}), after algebraic manipulations, we have
\begin{equation}
\label{suf1}
 \E \left \{\sum\nolimits_{j } \mu_{ij}^{\left( {d ,m} \right)} (t) +   \mu_{i,\text{pr}}^{\left( {d ,m} \right)} (t)  - \right. \left.  \sum\nolimits_{j } \mu_{ji}^{\left( {d ,m} \right)} (t)- \mu_{\text{pr},i}^{\left( {d ,m} \right)} (t)  \right\}
\geq \lambda_i^{(d,m)} \!+ \epsilon .
\end{equation}
By applying standard LDP analysis \cite{Neely_book2}, strong network stability  (\ie $\{\lambda_i^{(d,\phi,m)} \}$ in the interior of the capacity region)
follows.

\section{Proof of Theorem \ref{thm: network_capacity_region2}}
\label{appendix: proof of theorem 2}

Let the \emph{Lyapunov drift}  \cite{Neely_book2} for the queue backlogs of the network be defined as
\begin{equation}
\label{eq_Lypunov_drift}
\Delta({{\cal H}}(t))\triangleq\frac{1}{2} \! \sum\nolimits_{i,(d ,m)} \! {\mathbb{E} \! \left[ \! {\left. {{{\left(\! {Q_i^{\left( {d ,m} \right)} \!\left( {t + 1} \right)} \!\right)}^2} \!\!-\! {{\left( \!{Q_i^{\left( {d,m} \right)}\!\left( t \right)} \! \right)}^2}} \right|{\cal H}\left( t \right)} \right]}.\notag
\end{equation}
After standard LDP algebraic manipulations on \eqref{eq_queueing_dynamic} (see Ref. \cite{Neely_book2}), we have
\begin{align}
\label{eq_Lypunov_drift_bound1}
&\!\!{\!\!\!\Delta ({\cal {H}}(t))\! +\! V\mathbb{E}\{\left.h(t)\right|{\cal H}(t)\}\! \le\! NB\! + \!\! \sum\nolimits_{i,(d,m)}\!  {\lambda _i^{\left( {d,m} \right)}Q_i^{\left( {d,m} \right)}\! \left( t \right) \!}} \notag\\
&\!\!\!{- \sum\nolimits_{i } {\mathbb{E}\left\{ {\!\left. {Z_i^{\text{pr}} \! \left( t \right)\! - \!Vh_i^{\text{pr}}(t) + Z_i^{\text{tr}} \!\left( t \right) \! - \!Vh_i^{\text{tr}}(t) } \right|{\cal H}\left( t \right)} \right\}},}
\end{align}
where, with $r_{\min} \triangleq \min_{m}\{r^{(m)}\}$ and $\xi_{\max} \triangleq \max_{m}\{\xi^{(m)}\}$, we define 
\begin{align}
&{B \triangleq \frac{1}{2}\max\nolimits_{i}\left\{{\left( {\mathop {\max }\nolimits_{j,{\bf s}:j \ne i,{\bf{s}} \in S} \left\{ {{R_{ij,K_i^{\text{tr}}}}\left( {\bf{s}} \right)} \right\} + \left.{R_{i,K_i^{\text{pr}}}}\right/{r_{\min }}} \right)^2}\right.} \nonumber\\
&\qquad{\left.+ {\left( {\mathop {\max }\nolimits_{{\bf{s}} \in S} \left\{ {\sum\nolimits_{j:j \ne i} {{R_{ji,K_j^{\text{tr}}}}\left( {\bf{s}} \right)} } \right\} + \left.{\xi _{\max }}{R_{i,K_i^{\text{pr}}}}\right/{r_{\min }} + {A_{\max }}} \right)^2}\right\}\!,}\nonumber
\end{align}
\begin{align}
&  {Z^{\text{pr}} _{i}\!(t)   \triangleq \displaystyle \sum\nolimits_{(d ,m)} \! \mu_{i,\text{pr}}^{(d ,m)}(t) \! \left[  Q_i^{(d ,m)}(t) -  \xi^{(m+1)} Q_i^{(d ,m+1)}(t)\right ]}, \notag\\
& { Z^{\text{tr}} _{i}(t)  \triangleq \displaystyle \sum\nolimits_{u=1}^{N-1}{\sum\nolimits_{(d ,m)}  \! \mu_{iq_{i,u}}^{(d ,m)}(t) \! \left[  Q_i^{(d ,m)}(t) -   Q_{q_{i,u}}^{(d ,m)}(t) \right ]  }},   \label{eq_Z_tr_def}\\
&h_i^{\text{pr}}(t) \triangleq \!\sum\nolimits_{k =0}^{K_i^{\text{pr}}} {w_{i,k}^{\text{pr}}y_{i,k}^{\text{pr}}\left( \tau  \right)},\ h_i^{\text{tr}}(t) \triangleq \sum\nolimits_{k =0}^{K_i^{\text{tr}}} {w_{i,k}^{\text{tr}}y_{i,k}^{\text{tr}}\left( \tau  \right)}. \notag
\end{align}

\begin{lem}
\label{lem: maximize metric}
Among the algorithms using single-copy routing, the 
DWCNC algorithm, in each timeslot $t$, 
maximizes $\mathbb{E}\{\left.Z_i^{\text{\emph{tr}}}\left( t \right)\! -\! Vh_i^{\text{\emph{tr}}}(t)\right|{\cal H}(t)\}$ subject to (\ref{eq_rate_vs_group_rate2})-(\ref{eq_group_rate}) and $\mathbb{E}\{\left.Z_i^{\text{\emph{pr}}}\left( t \right)\! -\! Vh_i^{\text{\emph{pr}}}(t)\right|{\cal H}(t)\}$ subject to (\ref{chain})-(\ref{ratepr}).\hfill $\square$

\end{lem}

\begin{proof}
See Appendix \ref{appendix: proof}.
\end{proof}

Lemma \ref{lem: maximize metric} implies that the right hand side of \eqref{eq_Lypunov_drift_bound1} under DWCNC is no larger than
the corresponding expression under the optimal stationary randomized policy (characterized in Theorem \ref{thm: network_capacity_region}) that supports $({\bm \lambda}+\epsilon {\bf 1})\in \Lambda$ and
achieves average cost $\overline h^*({\bm \lambda}+\epsilon {\bf 1})$:
\begin{align}
\label{eq_drift_plus_penalty}
&{\Delta ({\cal{H}}(t)) \!+\! V\mathbb{E}\{\left.h(t)\right|{\cal H}(t)\} \le NB +\!\! \sum\nolimits_{i,(d ,m)} \! {\lambda _i^{\left( {d ,m} \right)}Q_i^{\left( {d ,m} \right)}\! \left( t \right)\!} }\notag\\
&{\quad+\sum\nolimits_{i } {\mathbb{E}\left[ {\left. {\!Z_i^{*\text{pr}}\!\left( t \right)\! - \!Vh_i^{*\text{pr}}(t)\! + Z_i^{*\text{tr}}\!\left( t \right)\! - \!Vh_i^{*\text{tr}}(t) } \right|{\cal H}\!\left( t \right)} \right]}\! }\notag\\
& { \le NB+V{\overline h^*}\!\!\left( { {\bm \lambda}\! +\! \epsilon {\bf{1}}} \right)\!-\! \epsilon \sum\nolimits_i \! {\sum\nolimits_{\left( {d ,m} \right)} \!{Q_i^{\left( {d ,m} \right)}\!\left( t \right).} }}
\end{align}


Finally, we can use the theoretical result in \cite{Neely_prob_1} for the proof of network stability and average cost convergence with probability $1$. Note that the following bounding conditions are satisfied
in the network system:
\begin{enumerate}
\item The second moment of $\mathbb{E}\{(h(t))^2\}$ is upper bounded by $ \sum\nolimits_{i}{(w_{i,K_i^{\text{tr}}}^{\text{tr}}+w_{i,K_i^{\text{pr}}}^{\text{pr}})}$ and therefore satisfies $\sum\nolimits_{\tau=0}^{\infty}{\left.\mathbb{E}\{(h(t))^2\}\right/\tau^2}<\infty$.
\item We have $\mathbb{E}\{\left.h(t)\right|\mathcal H(t)\}$ lower bounded for all $\mathcal H (t)$ and $t$: $\mathbb{E}\{\left.h(t)\right|\mathcal H(t)\}\ge 0$.
\item The conditional fourth moment of queue length change is upper bounded for all $\mathcal H (t)$, $t$, $i$ and $(d,m)$:
    \begin{align}
    &\Scale[0.97]{\mathbb{E}\!\left\{\!\!\left.\left(Q_i^{(d,m)}(t+1)-Q_i^{(d,m)}(t)\right)^4\right|\!\mathcal H(t)\!\right\}\!\le\!\max\limits_i\!\left\{\!\!\left[ {\mathop {\max }\limits_{{\bf{s}} \in S} \left\{ {\sum\limits_{j:j \ne i} {{R_{ji,K_j^{\text{tr}}}}\left( {\bf{s}} \right)} } \!\right\} \!+\! \frac{{\xi _{\max }}{R_{i,K_i^{\text{pr}}}}}{r_{\min }} \!+\! {A_{\max }}} \right]^4\!\right\}\!\!.\notag}
    \end{align}
\end{enumerate}
With the above three conditions satisfied, based on the derivations in \cite{Neely_prob_1}, Eq. \eqref{eq_drift_plus_penalty} leads to the network stability \eqref{que} and average cost \eqref{eq_ave_cost} convergence with probability $1$ of DWCNC.


\section{Proof of Lemma \ref{lem: maximize metric}}
\label{appendix: proof}

Regarding the processing decisions, since the {\em computing channel} is always known, maximizing $\mathbb{E}\{\left.Z_i^{\text{pr}}\left( t \right) - Vh_i^{\text{pr}}(t)\right|{\cal H}(t)\}$  is equivalent to maximizing $Z_i^{\text{pr}}\left( t \right) - Vh_i^{\text{pr}}(t)$. And the maximization of $Z_i^{\text{pr}}\left( t \right) - Vh_i^{\text{pr}}(t)$ subject to (\ref{chain})-(\ref{ratepr}) can be directly achieved by the choice of commodity $(d,m)$, resource allocation $k$, and flow rate $\mu_{i,\text{pr}}^{(d,m)}(t)$ described by the local processing decisions of DWCNC in Sec. \ref{subsec: alg_description}.

With respect to the transmission decisions, 
it follows by plugging \eqref{eq_rate_vs_group_rate2} into \eqref{eq_Z_tr_def} that
\begin{equation}
\label{eq_Z_tr2}
{Z_i^{\text{tr}}\left( t \right) =\sum\nolimits_{\left( {d,m} \right)} \! {\sum\nolimits_{n = 1}^{N - 1} {\! { {\sum\nolimits_{u = n}^{N - 1} \! {\mu _{i{q_{i,u}},n}^{\left( {d,m} \right)}\!\left( t \right)} \!\left[ {Q_i^{\left( {d,m} \right)}\!\left( t \right) - Q_{{q_{i,u}}}^{\left( {d,m} \right)}\!\left( t \right)} \right]} } } } }.
\end{equation}
Let $\chi_{i,\text{tr}}^{(d,m)}(t)$ be the fraction of the transmission time allocated to the transmission of commodity $(d,m)$ in timeslot $t$,
and let $\eta_{ij,n}^{(d,m)}(t)$ be the fraction of the transmitted commodity $(d,m)$ in the $n$-th partition that is retained by node $j$, with $n\le q_{i,{\bf S}(t)}^{-1}(j)$. Then, assuming single-copy routing, it follows from \eqref{eq_group_rate} that
\begin{align}
\label{eq_mu_group_decompose}
&{\mu _{i{q_{i,u}},n}^{\left( {d,m} \right)}\!\left( t \right) = \chi _{i,\text{tr}}^{(d,m)}\!(t)\eta _{i{q_{i,u}},n}^{(d,m)}\!(t)\! \left[ {{R_{i{q_{i,n}},k}}\!\left( {{\bf{S}}\!\left( t \right)} \right) - {R_{i{q_{i,n - 1}},k}}\!\left( {{\bf{S}}\!\left( t \right)} \right)} \right]},\ \ \forall i,\ t,\\ \label{eq_beta}
&\sum\nolimits_{\left( {d ,m} \right)} {\chi _{i,\text{tr}}^{(d ,m)}(t)}  \le 1,\ \ \forall i,\ t, \\
\label{eq_eta}
&\sum\nolimits_j {\eta _{ij,n}^{(d ,m)}(t)}  \le 1,\ \ \forall i,\ t,\ (d,m).
\end{align}
Plugging \eqref{eq_mu_group_decompose} into \eqref{eq_Z_tr2} and taking the expectation conditioned on ${\cal H}(t)$ and $\{y_{i,k}^{\text{tr}}(t)=1\}$, it follows that
\begin{align}
\label{eq_upper_bound1}
&\Scale[0.94]{{\mathbb{E}\left\{\left.Z_i^{\text{tr}}\left( t \right)\right|{\cal H}(t),y_{i,k}^{\text{tr}}(t)=1\right\}}}\notag\\
&\Scale[0.94]{ \buildrel (a) \over \le\! \sum\limits_{\left( {d ,m} \right)} \sum\limits_{n = 1}^{N - 1} {\mathbb{E}\left\{ {\left.\chi _{i,\text{tr}}^{(d ,m)}(t){\left[ {{R_{i{q_{i,n}},k}}\left( {{\bf{S}}\left( t \right)} \right) - {R_{i{q_{i,n - 1}},k}}\left( {{\bf{S}}\left( t \right)} \right)} \right]}\sum\limits_{u = n}^{N - 1} {\eta _{i{q_{i,u}},n}^{(d ,m)}(t)W_{i{g_{i,u}}}^{\left( {d ,m} \right)}\left( t \right)} \right|{\cal H}\left( t \right),y_{i,k}^{\text{tr}}(t)=1} \right\}}}\notag\\
& \Scale[0.94]{\buildrel (b) \over \le \sum\limits_{\left( {d ,m} \right)} {\sum\limits_{n = 1}^{N - 1} {\mathbb{E}\left\{ {\left.\chi _{i,\text{tr}}^{(d ,m)}(t){\left[ {{R_{i{q_{i,n}},k}}\left( {{\bf{S}}\left( t \right)} \right) - {R_{i{q_{i,n - 1}},k}}\left( {{\bf{S}}\left( t \right)} \right)} \right]}\mathop {\max }\limits_{j\in \Omega_{i,n}({\bf S}(t))}\!\! \left\{ {W_{ij}^{\left( {d ,m} \right)}\left( t \right)} \right\} \right|{\cal H}\left( t \right),y_{i,k}^{\text{tr}}(t)=1}\! \right\}}}}\notag\\
&\Scale[0.94]{\!\buildrel (c)\over = \sum\nolimits_{\left( {d,m} \right)} {\mathbb{E}\left\{ {\left. {\chi _{i,\text{tr}}^{(d,m)}(t)} \right|{\cal H}\left( t \right)}, y_{i,k}^{\text{tr}}(t)=1 \right\}}{\sum\nolimits_{n = 1}^{N - 1} {\mathbb{E}\left\{ {{\max \nolimits_{j \in {\Omega _{i,n}}({\bf{S}}(t))}}\!\!\left\{ {W_{ij}^{\left( {d,m} \right)}\left( t \right)} \right\}} \right.} }} \notag \\
&\Scale[0.94]{\qquad\qquad\times \left. {\left. {\left[ {{R_{i{q_{i,n}},k}}\left( {{\bf{S}}\left( t \right)} \right) - {R_{i{q_{i,n - 1}},k}}\left( {{\bf{S}}\left( t \right)} \right)} \right]} \right|{\cal H}\left( t \right),y_{i,k}^{\text{tr}}(t)=1} \right\}}\notag\\
&\Scale[0.94]{\buildrel (d) \over \le \mathop {\max }\limits_{\left( {d ,m} \right)} \left\{ {\sum\limits_{n = 1}^{N - 1} \mathbb{E}\left\{{\left. {\left[ {{R_{i{q_{i,n}},k}}\left( {{\bf{S}}\left( t \right)} \right) - {R_{i{q_{i,n - 1}},k}}\left( {{\bf{S}}\left( t \right)} \right)} \right]} {\mathop {\max }\limits_{j \in {\Omega _{i,n}}({\bf{S}}(t))} \left\{ {W_{ij}^{\left( {d ,m} \right)}\left( t \right)} \right\}} \right|{\cal H}\left( t \right),y_{i,k}^{\text{tr}}(t)=1} \right\} }\right\}}\notag\\
&\Scale[0.94]{ \buildrel (e) \over = \mathop {\max }\limits_{\left( {d,m} \right)} \left\{ {W_{i,k,\text{tr}}^{\left( {d,m} \right)}\left( t \right)}\right\}}.
\end{align}
In \eqref{eq_upper_bound1}, inequality $(a)$ follows from the definition of $W_{ij}^{(d,\phi,m)}(t)$; 
inequality $(b)$ follows from \eqref{eq_eta}; 
equality $(c)$ holds because, given ${\cal H}(t)$ and $\{y_{i,k}^{\text{tr}}(t)=1\}$, the values of $ {R_{i{q_{i,n}},k}}\left( {{\bf{S}}\left( t \right)} \right)$ and $\mathop {\max }\nolimits_{j\in \Omega_{i,n}({\bf S}(t))} \{ {W_{ij}^{\left( {d,m} \right)}\left( t \right)} \}$ are determined by ${\bf S}(t)$ and therefore are independent from ${\chi _{i,\text{tr}}^{(d ,m)}(t)}$; inequality $(d)$ follows from \eqref{eq_beta}; equality $(e)$ follows from the definition of $W_{i,k,\text{tr}}^{\left( {d,m} \right)}\left( t \right)$ in \eqref{Wp}.

Finally, taking expectation over $y_{i,k}^{\text{tr}}(t)$ on \eqref{eq_upper_bound1}, we further have
\begin{align}
\label{eq_upper_bound3}
&\Scale[1]{\mathbb{E}\left\{ {\left. {Z_i^{\text{tr}}\left( t \right) - Vh_i^{\text{tr}}\left( t \right)} \right|{\cal H}(t)} \right\}}\Scale[1]{\ \le \sum\nolimits_{k = 0}^{K_i^{\text{tr}}} {\left[ {{{\max }_{\left( {d,m} \right)}}\left\{ {W_{i,k,\text{tr}}^{\left( {d,m} \right)}\left( t \right)} \right\} - Vw_{i,k}^{\text{tr}}} \right]\Pr \left\{ {y_{i,k}^{\text{tr}}(t) = 1} \right\}}}\notag\\
&\Scale[1]{\ \buildrel (f) \over \le {\max \nolimits_{k,\left( {d,m} \right)}}\left\{ {W_{i,k,\text{tr}}^{\left( {d,m} \right)}\left( t \right)}  - Vw_{i,k}^{\text{tr}} \right\}},
\end{align}
where $(f)$ follows due to the fact that $\sum\nolimits_{k = 0}^{K_i^{\text{tr}}} {\Pr \left\{ {y_{i,k}^{\text{tr}}(t) = 1} \right\}}=1$.

In \eqref{eq_upper_bound1} and \eqref{eq_upper_bound3}, the upper bounds $(a)$ and $(b)$ can be achieved by implementing step \ref{step: forwarding_decision} of the local transmission decisions of DWCNC;
the upper bound $(d)$, $(e)$, and $(f)$ can be achieved by implementing step \ref{step: compute_trans_utility_weight} and \ref{step: choose_commodity_resource} of the local transmission decisions of DWCNC.
This concludes the proof of Lemma \ref{lem: maximize metric}.


\ifCLASSOPTIONcaptionsoff
  \newpage
\fi


\begin{thebibliography}{}

\bibitem{dwcnc2017}
H. Feng, J. Llorca, A. M. Tulino, A. F. Molisch, 
``On the delivery of augmented information services over wireless computing networks," 
\emph{IEEE International Conference on Communications (ICC)}, Paris, June 2017, pp. 1-7.


\bibitem{fxbook} Marcus Weldon,
``The Future X Network,''
\emph{CRC Press}, October 2015.


\bibitem{vnf} L. Lewin-Eytan, J. Naor, R. Cohen, D. Raz,
``Near Optimal Placement of Virtual Network Functions,''
\emph{IEEE INFOCOM}, 2015.


\bibitem{csdp} M. Barcelo, J. Llorca, A. M. Tulino, N. Raman,
``The Cloud Servide Distribution Problem in Distributed Cloud Networks,"
\emph{IEEE ICC}, Sept. 2015.

\bibitem{infocom17} H. Feng, J. Llorca, A. M. Tulino, D. Raz, A. F. Molisch,
``Approximation Algorithms for the Network Service Distribution Problem,"
{\em IEEE INFOCOM}, June 2017.

\bibitem{infocom} H. Feng, J. Llorca, A. M. Tulino, A. F. Molisch,
``Dynamic Network Service Optimization in Distributed Cloud Networks,"
{\em IEEE INFOCOM SWFAN Workshop}, Sept. 2016.

\bibitem{icc} H. Feng, J. Llorca, A. M. Tulino and A. F. Molisch,
``Optimal Dynamic Cloud Network Control,"
\emph{IEEE ICC}, Sept. 2016. 

\bibitem{fog} M. Chiang and T. Zhang, ``Fog and IoT: An Overview of Research Opportunities'', \emph{IEEE Internet of Things Journal}, vol. {3}, no. {6}, pp. {854-864}, Dec. {2016}.

\bibitem{cloudlet} M. Satyanarayanan, P. Bahl, R. Caceres and N. Davies, ``The Case for VM-Based Cloudlets in Mobile Computing'', \emph{IEEE Pervasive Computing}, vol. {8}, no. {4}, pp. {14-23}, Oct.-Dec. {2009}.

\bibitem{Neely_book2} M. J. Neely, ``Stochastic Network Optimization with Application to Communication and Queueing Systems",
{\em Synthesis Lectures on Communication Networks},  {Morgan \& Claypool}, {2010}.

\bibitem{DIVBAR_Neely} M. J. Neely, ``Optimal Backpressure Routing for Wireless Networks with Multi-Receiver Diversity",
{\em Ad Hoc Networks}, vol. {7}, pp. {862--881}, {July 2009}.

\bibitem{DIVBAR_Molisch} H. Feng and A. F. Molisch, ``Diversity Backpressure Scheduling and Routing with Mutual Information Accumulation in Wireless Ad-hoc Networks",
{\em IEEE Transactions on Information Theory}, vol. 62, no. 12: pp. 7299-7323, Dec. 2016.





\bibitem{Shamai_Steiner}
S. Shamai and A. Steiner. ``A Broadcast Approach for A Single-user Slowly Fading MIMO Channel." \emph{IEEE Transactions on Information Theory}, vol. 49, no. 10: pp. 2617-2635, Oct. 2003.

\bibitem{Tulino_et_al_2014}
A. M. Tulino, G. Caire and S. Shamai. ``The Broadcast Approach for The Sparse-input Random-sampled MIMO Gaussian Channel," \emph{IEEE ISIT}, Aug. 2014.

\bibitem{el2011network}
A.~El~Gamal and Y.-H. Kim, ``Network Information Theory,'' \emph{Cambridge University Press}, 2011.

\bibitem{molisch_book}
A. F. Molisch, ``Wireless communications,'' 2nd ed., \emph{John Wiley $\&$ Sons}, 2012.


\bibitem{Neely_prob_1} M. J. Neely, ``Queue Stability and Probability $1$ Convergence via Lyapunov Optimization,''
\emph{arXiv preprint:1008.3519}, {2010}.








\end{thebibliography}
\end{document}